\documentclass[onefignum,onetabnum]{siamart171218}

\usepackage{amssymb}
\usepackage{amsmath}
\usepackage{amsfonts}
\usepackage{enumerate} 
\usepackage{caption,subcaption}
\usepackage{hyperref} 

\usepackage{lipsum}
\usepackage{amsfonts}
\usepackage{graphicx}
\usepackage{epstopdf}
\usepackage{algorithmic}
\ifpdf
  \DeclareGraphicsExtensions{.eps,.pdf,.png,.jpg}
\else
  \DeclareGraphicsExtensions{.eps}
\fi

\newsiamremark{remark}{Remark}
\newsiamremark{hypothesis}{Hypothesis}
\crefname{hypothesis}{Hypothesis}{Hypotheses}
\newsiamthm{claim}{Claim}

\headers{GS of quantum droplets}{Wei Liu, Limin Xu}

\title{Ground-state solution of quantum droplets in Bose-Bose mixtures\thanks{Submitted to the editors DATE.
\funding{This work was funded by the NSFC grant 12571448 and the Innovation Research Foundation of NUDT (202402-YJRC-XX-002).}}}

\author{Wei Liu\thanks{College of Science, National University of Defense Technology, Changsha 410073, China (\email{wl@nudt.edu.cn}).}
\and Limin Xu\thanks{Corresponding author. Institute for Theoretical Sciences, Westlake University, Hangzhou, 310030, China  (\email{xulimin@westlake.edu.cn})}}

\usepackage{amsopn}

\makeatletter
\newcommand*{\addFileDependency}[1]{
  \typeout{(#1)}
  \@addtofilelist{#1}
  \IfFileExists{#1}{}{\typeout{No file #1.}}
}
\makeatother

\ifpdf
\hypersetup{
  pdftitle={Ground-state solution of quantum droplets in Bose-Bose mixtures},
  pdfauthor={W. Liu, L. Xu}
}
\fi

\begin{document}

\maketitle

\begin{abstract} 
In this paper, we present a systematic study on the ground state computation of quantum droplets in homonuclear Bose-Bose mixtures, governed by the extended Gross-Pitaevskii equations (eGPEs) with Lee-Huang-Yang (LHY) corrections. This model captures the formation of self-bound droplets stabilized by the delicate balance between the attractive mean-field interaction and the repulsive quantum fluctuations. We formulate dimensionless energy functionals for both the general two-component system and the reduced single-component density-locked model. To compute the ground states efficiently, we adapt and benchmark various gradient flow discretization schemes, identifying a backward-forward sine-pseudospectral scheme based on the gradient flow with Lagrange multiplier method (GFLM-BFSP) as the robust solver for our simulations. Utilizing this method, we report three main numerical observations: (i) the density-locked model is quantitatively validated as a reliable approximation for ground state properties; (ii) the dimension-dependent convergence rates of the Thomas-Fermi approximation are established in the strong-coupling regime; and (iii) the critical particle number for self-binding in free space is numerically determined, providing a precise correction to the analytical prediction by Petrov [Phys. Rev. Lett. 115, 155302 (2015)]. 
\end{abstract}




\section{Introduction}
In 2015, Petrov \cite{petrov_quantum_2015,petrov_ultradilute_2016} predicted the existence of self-bound quantum droplets in binary Bose-Einstein condensates (BECs). According to mean-field theory, a mixture is expected to collapse when the inter-species attraction dominates the average intra-species repulsion. However, Petrov demonstrated that the system can be stabilized by an effective repulsion arising from the first beyond-mean-field correction, known as the Lee-Huang-Yang (LHY) correlation energy, thus forming stable droplets even in free space. While initially proposed for binary mixtures in free space, this concept has since been generalized to confined geometries and dipolar condensates. To date, quantum droplets have been experimentally observed in various systems: (1) dipolar gases of Dy \cite{ferrier-barbut_observation_2016,ferrier2016liquid,kadau2016observing,schmitt2016self,wenzel2017striped} and Er \cite{chomaz2016quantum}; (2) homonuclear two-component mixtures of $^{39}$K in free space \cite{semeghini_self-bound_2018} as well as under one- \cite{cabrera2018quantum} or two-dimensional confinement \cite{cheiney_bright_2018}; and (3) heteronuclear mixtures such as $^{41}$K--$^{87}$Rb \cite{derrico_observation_2019} and $^{23}$Na--$^{87}$Rb \cite{guo2021lee}.

Compared to traditional BECs, quantum droplets in Bose-Bose mixtures exhibit several exotic properties. The first is the mechanism of self-evaporation \cite{ferioli_dynamical_2020,petrov_quantum_2015}. Unlike conventional condensates, an excited droplet tends to release particles into the background gas, while the remaining atoms relax asymptotically to the ground state of the droplet core (provided the particle number remains within a stable range). Notably, this self-evaporation is absent in dipolar quantum droplets \cite{baillie2017collective} and in the breathing mode of one-dimensional two-component droplets \cite{tylutki2020collective}. The second key property is density locking (or density balancing) \cite{ferioli_dynamical_2020,flynn_quantum_2023,petrov_quantum_2015}. In the ground state, the density ratio between the two components is fixed to a specific value determined by the interaction strengths; any excess atoms of either component are expelled from the droplet core to the background cloud. These unique characteristics render quantum droplets ideal platforms for benchmarking quantum many-body theories against experimental observations \cite{ancilotto_self-bound_2018}. Consequently, the rapid experimental progress and these intriguing properties have stimulated a wave of theoretical and numerical investigations.

The macroscopic properties of quantum droplets in Bose-Bose mixtures are typically described by the extended Gross-Pitaevskii equation (eGPE), which is a nonlinear Schr\"{o}dinger-type equation with competing nonlinearities, derived by incorporating the Lee-Huang-Yang (LHY) correlation term into the standard Gross-Pitaevskii energy functional and subsequently applying the variational principle \cite{petrov_quantum_2015}.
Physically, the LHY correction arises from quantum fluctuations and provides a repulsive energy contribution proportional to $n^{5/2}$ (where $n$ is the density). This repulsion is crucial for stabilizing the system against the collapse induced by the attractive inter-species mean-field interaction, which scales as $n^2$.
The computation of the ground state for the eGPE is of fundamental importance, as it characterizes the equilibrium density profiles and stability regions of self-bound droplets---localized states that maintain a finite size in free space without any external confinement.
Given the structural similarity between the eGPE and the standard GPE, the efficient numerical algorithms established for traditional BECs provide a natural starting point for our investigation.

Extensive numerical studies have been devoted to computing the ground states of BECs based on the Gross-Pitaevskii theory; see, e.g., \cite{ALT2017JCP,BC2013KRM,bao2006efficient,bao2004computing,Cances10,CDLX2023JCP,CST2000PRE,DP2017SISC,faou2018convergence,HJ2025SIREV}. Gradient flow-based methods, also known as imaginary-time propagation methods, are among the most widely used approaches for computing the ground states. Seminal works include the continuous normalized gradient flow (CNGF) introduced by Bao and Du \cite{bao2004computing}, and its efficient discretization using sine-pseudospectral methods by Bao, Chern, and Lim \cite{bao2006efficient}. Recently, Liu and Cai \cite{liu2021normalized} proposed the gradient flow with Lagrange multiplier (GFLM) method to achieve more flexible numerical discretization by eliminating inherent temporal splitting errors in traditional normalized gradient flow methods. 
In contrast, numerical simulations of the eGPE for quantum droplets remain comparatively limited.
Existing studies have primarily employed imaginary-time propagation \cite{ferioli_dynamical_2020} or Runge-Kutta methods \cite{flynn_quantum_2023}, often relying on lower-order spatial discretizations.
To date, a systematic evaluation of high-order spectral gradient flow algorithms has not been thoroughly conducted for quantum droplets, particularly in the challenging regime where mean-field attraction competes with LHY repulsion.

In this paper, we aim to bridge this gap by establishing a systematic numerical framework for quantum droplets and exploring their ground state properties. Our contributions are threefold. First, we provide a rigorous formulation of the dimensionless eGPEs and derive the reduced single-component density-locked model, clarifying the connections between different effective models. Second, we adapt and benchmark various gradient flow algorithms for this system. Through extensive comparisons, we identify the GFLM-BFSP scheme as the optimal solver, demonstrating that explicitly handling the normalization constraint is crucial for correcting time-splitting errors in the presence of LHY interactions. Third, utilizing this efficient tool, we investigate physical properties that are difficult to access analytically. We validate the accuracy of the density-locked approximation, quantify the dimension-dependent convergence rates of the Thomas-Fermi approximation (TFA), and numerically determine the precise critical particle number $N_c$ for self-binding in free space, providing a correction to previous analytical estimates based on Gaussian ansatz.

The paper is organized as follows. In Section~\ref{sec:eGPE}, we derive the dimensionless eGPE and discuss its reduction to lower dimensions and the density-locked model. In Section~\ref{sec:Numerical methods}, we detail the CNGF and GFLM methods along with their spectral discretizations, and provide approximate initial data for different regimes. Section~\ref{sec:Numerical results} presents the numerical results, including the performance comparison of algorithms, the validation of the density-locked model, and the investigation of physical properties such as $N_c$ and TFA convergence. Finally, a summary is given in Section~\ref{sec:Summary}.

\section{Extended Gross-Pitaeskii equation}\label{sec:eGPE}
\subsection{Energy functional and extended Gross-Pitaevskii equation}
The energy functional of quantum droplets in Bose-Bose mixtures is given by \cite{ancilotto_self-bound_2018,derrico_observation_2019,ferioli_dynamical_2020,flynn_quantum_2023,petrov_quantum_2015}
\begin{align}\label{eq: energy functional}
E(\psi_1,\psi_2)&=  \sum_{j=1}^2 \int_{\mathbb{R}^3} \left[\frac{\hbar^2}{2 m_j}\left|\nabla \psi_j\right|^2+V_j(\mathbf{x}) n_j\right] \mathrm{d} \mathbf{x} 
+\frac{1}{2} \sum_{j,l=1}^2 g_{jl} \int_{\mathbb{R}^3} n_j n_l \,\mathrm{d} \mathbf{x}\\
&\quad + \int_{\mathbb{R}^3} \mathcal{E}_{\mathrm{LHY}}\left(n_1, n_2\right) \mathrm{d} \mathbf{x}, \nonumber
\end{align}
where $\hbar$ is the reduced Planck constant, $\psi_j(\mathbf{x},t)$ is the wave function of the $j$-th component with density $n_j(\mathbf{x},t)=\left|\psi_j(\mathbf{x},t)\right|^2$, and $m_j$ is the mass of particles in the $j$-th component ($j=1,2$). The external trapping potentials are denoted by $V_{j}(\mathbf{x})=\frac{1}{2}m_j\left(\omega_{x,j}^2x^2+\omega_{y,j}^2y^2+\omega_{z,j}^2z^2\right)$. The intra- and inter-species coupling constants are defined as $g_{jl}=\frac{4\pi\hbar^2a_{jl}}{m_{jl}}$,  where $a_{jj}>0$ and $a_{12}<0$ are the repulsive intra-species and attractive inter-species scattering lengths, respectively. $m_{jl}$ are defined as $m_{jj}=m_j$ and $m_{12}={2m_1m_2}/{(m_1+m_2)}$. 

The LHY correlation energy density, which accounts for quantum fluctuations, is expressed as
\begin{equation}
\mathcal{E}_{\mathrm{LHY}}\left(n_1, n_2\right) =\frac{8}{15 \pi^2}\left(\frac{m_1}{\hbar^2}\right)^{3 / 2}\left(g_{11} n_1\right)^{5 / 2} f\left(\frac{m_2}{m_1}, \frac{g_{12}^2}{g_{11} g_{22}}, \frac{g_{22} n_2}{g_{11} n_1}\right).
\end{equation}
The dimensionless function $f(z,u,x)>0$ is defined by the integral \cite{minardi_effective_2019}
\begin{equation*}
	f(z, u, x)=\frac{15}{32} \int_0^{\infty} k^2 \mathcal{F}(k, z, u, x)\, \mathrm{d} k,
\end{equation*}
where the integrand $\mathcal{F}$ takes the compact form
\begin{equation*}
\mathcal{F}(k, z, u, x) = \sqrt{\lambda_+(k)} + \sqrt{\lambda_-(k)} -\frac{1+z}{2 z} k^2-(1+x)+\frac{1}{k^2}\left[1+x^2 z+4 u x \frac{z}{1+z}\right],
\end{equation*}
with the dimensionless squared eigenenergies $\lambda_\pm(k)$, obtained by diagonalizing the Bogoliubov Hamiltonian for the binary mixture, defined as
\begin{equation*}
\lambda_\pm(k) =  \frac{1}{2}\left[k^2\left(1\!+\!\frac{x}{z}\right)\!+\!\frac{k^4}{4}\left(1\!+\!\frac{1}{z^2}\right)\right] 
 \pm \sqrt{\frac{1}{4}\left[\left(k^2\!+\!\frac{k^4}{4}\right)\!-\!\left(\frac{x}{z} k^2\!+\!\frac{k^4}{4 z^2}\right)\right]^2\!+k^4u \frac{x}{z}}.
\end{equation*}

In the regime relevant to quantum droplets, where the inter-species attraction nearly compensates the intra-species repulsion, the function $f(z, u, x)$ can be well approximated by $f\left(z,u,x\right) \simeq f\left(z,1,x\right)\simeq (1+z^{3/5} x)^{5/2}$ \cite{minardi_effective_2019}. Consequently,
\begin{equation}\label{eq:eLHYdensity}
    \mathcal{E}_{\mathrm{LHY}}\left(n_1, n_2\right)\simeq \frac{256\sqrt{\pi}\hbar^2}{15}\left(\frac{a_{11}}{m_1^{2/5}}|\psi_1|^2+\frac{a_{22}}{m_2^{2/5}}|\psi_2|^2\right)^{5/2}.
\end{equation}
With the explicit expression \eqref{eq:eLHYdensity}, the coupled eGPE can be then derived by utilizing the variational derivations of the total energy $E(\psi_1, \psi_2)$ with respect to $\psi_1^*$ and $\psi_2^*$ as
\begin{align}\label{eq: eGPE_system}
    \mathrm{i} \hbar \partial_t \psi_j 
    &=\!\left[-\frac{\hbar^2}{2 m_j} \nabla^2+V_j(\mathbf{x})+\sum_{l=1}^2 g_{jl}|\psi_l|^2 
    +\frac{128\sqrt{\pi}\hbar^2a_{jj}}{3m_j^{2/5}}  \left(\sum_{l=1}^2 \frac{a_{ll}}{m_l^{2/5}}|\psi_l|^2\!\right)^{\!3/2}\right] \psi_j, 
\end{align}
where the wave functions are normalized such that $\int_{\mathbb{R}^3}|\psi_j(\mathbf{x}, t)|^2 \mathrm{d} \mathbf{x}=N_j$ for $j=1,2$, with $N_j$ the number of particles in $j$-th component.

\subsection{Non-dimensionalization}\label{sec:Non-dim}
To non-dimensionalize the eGPE \eqref{eq: eGPE_system} under the normalization condition, we introduce the following scaling:
\begin{equation}\label{eq: dimensionless scaling}
    \tilde{t}=\omega t,\;\; \tilde{\mathbf{x}}=\frac{\mathbf{x}}{x_s},\;\; \tilde{\psi}_j(\tilde{\mathbf{x}},\tilde{t})= \frac{x_s^{3/2}\psi_j(\mathbf{x},t)}{\sqrt{N}},\; j=1,2,\;\; \tilde{E}(\tilde{\psi}_1,\tilde{\psi}_2)=\frac{E(\psi_1,\psi_2)}{N\hbar\omega},
\end{equation}
where $x_s=\sqrt{\frac{\hbar}{m_1\omega}}$ is the characteristic length scale with $\omega>0$ a reference frequency to be determined later, and $N=N_1+N_2$ is the total number of particles. Substituting \eqref{eq: dimensionless scaling} into \eqref{eq: eGPE_system} and omitting all $\tilde{~}$ for simplicity, we obtain the dimensionless eGPE:
\begin{equation}\label{eq: dimensionless eGPE}
\mathrm{i}\partial_t\psi_j =\!\left[-\frac{\varepsilon_j}{2}\nabla^2+V_j(\mathbf{x})+\sum_{l=1}^2\alpha_{jl}|\psi_l|^2+\delta\beta_j\left(\sum_{l=1}^2\beta_{l}|\psi_l|^2\!\right)^{\!3/2}\right]\psi_j,\;\; j=1,2,
\end{equation}
where the dimensionless parameters are defined as follows:
\begin{align*}
        &\varepsilon_j=\frac{m_1}{m_j},\quad V_j(\mathbf{x})=\frac{1}{2}\sum_{\nu=x,y,z}\gamma_{\nu,j}^2\nu^2,\quad \gamma_{\nu,j}=\frac{\omega_{\nu,j}}{\omega\sqrt{\varepsilon_j}},\quad \nu \in \{x,y,z\},\quad j=1,2, \\
        &\alpha_{jl} = \frac{2\pi a_{jl}N(\varepsilon_j+\varepsilon_l)}{x_s},\quad 
        \beta_j=\frac{a_{jj}\varepsilon_j^{2/5}}{x_s},\quad j,l=1,2,\quad  \delta=\frac{128\sqrt{\pi}}{3}N^{3/2}.
\end{align*}
The dimensionless energy (per particle) is
\begin{align}\label{eq:dimensionlessEnergy}
E(\psi_1,\psi_2)&= \int_{\mathbb{R}^3} \left[ \sum_{j=1}^2\left(\frac{\varepsilon_j}{2}\left|\nabla \psi_j\right|^2+V_j(\mathbf{x})|\psi_j|^2\right) 
+\frac{1}{2} \sum_{j,l=1}^2 \alpha_{jl} |\psi_j|^2|\psi_l|^2\right. \\
&\quad + \left.\frac{2\delta}{5}\left(\sum_{l=1}^2\beta_{l}|\psi_l|^2\right)^{5/2}\right] \mathrm{d} \mathbf{x}, \nonumber
\end{align}
and the normalization condition of the dimensionless wave functions becomes
\begin{equation}\label{normalization for dimensionless eGPE2}
\int_{\mathbb{R}^3}|\psi_j(\mathbf{x}, t)|^2 \mathrm{d} \mathbf{x}=\frac{N_j}{N},\quad j=1,2.
\end{equation}
For systems with harmonic trapping potentials (i.e., $\omega_{\nu,j}>0$, $\nu\in\{x,y,z\}$, $j=1,2$), we choose $\omega = \min_{j=1,2}\{\omega_{x,j}, \omega_{y,j},\omega_{z,j}\}$ \cite{cappellaro_collective_2018}. For free-space droplets (i.e., $V_1(\mathbf{x})=V_2(\mathbf{x})=0$), we adopt the choice used in \cite{petrov_quantum_2015}:
\begin{equation}\label{eq:omega-free-space}
\omega=\frac{25\pi^2\hbar\,|\delta a|^3}{384m\,a_{11}a_{22}\left(\sqrt{a_{11}}+\sqrt{a_{22}}\right)^6}\quad
\mbox{with}\;\; \delta a = a_{12} + \sqrt{a_{11}a_{22}}.
\end{equation}

\begin{remark}
We note that almost all experiments studying the dynamics of quantum droplets in free space begin with the ground state of a BEC in a harmonic trap, which is subsequently turned off at $t=0$ to allow the system to evolve \cite{derrico_observation_2019,guo2021lee,semeghini_self-bound_2018}. Therefore, in such dynamical studies, it is also natural to choose $\omega$ as the minimal frequency of the initial harmonic trap used for preparation. 
\end{remark}


For compactness, particularly for the description of numerical methods, the dimensionless eGPE \eqref{eq: dimensionless eGPE} can be written in a vector form:
\begin{equation} \label{eq: compact eGPE}
    \mathrm{i}\frac{\partial \Psi}{\partial t}=-\frac{1}{2}\boldsymbol{\varepsilon}\odot\nabla^2\Psi+\mathbf{V}(\mathbf{x})\odot\Psi+\mathbf{A}(\Psi)\odot\Psi+\mathbf{B}(\Psi)\odot\Psi,
\end{equation}
where $\Psi = (\psi_1, \psi_2)^T$, $\boldsymbol{\varepsilon}=(\varepsilon_1,\varepsilon_2)^T$, and $\mathbf{V}(\mathbf{x})=(V_1(\mathbf{x}),V_2(\mathbf{x}))^T$. The nonlinear vector functions are defined as
\begin{subequations}\label{eq: nonlinear terms 3d}
\begin{align}
    &\mathbf{A}(\Psi)=(A_1(\Psi),A_2(\Psi))^T,\quad A_j(\Psi)=\sum_{l=1}^2\alpha_{jl}|\psi_l|^2,\quad j=1,2,\\
    &\mathbf{B}(\Psi)=(B_1(\Psi),B_2(\Psi))^T,\quad B_j(\Psi)=\delta\beta_j\left(\sum_{l=1}^2\beta_{l}|\psi_l|^2\right)^{3/2},\quad j=1,2.
\end{align}
\end{subequations}
Here, $\odot$ denotes the Hadamard (element-wise) product between two vectors, i.e., for $\mathbf{{U}}=(u_1, u_2)^{T}$ and $\mathbf{{V}}=(v_1, v_2)^{T}$, $\mathbf{U}\odot\mathbf{V}=(u_1v_1, u_2v_2)^{T}$.

\subsection{Reduction to lower dimensions under strongly anisotropic potentials}
Under strongly anisotropic external potentials, the 3D eGPE can be approximately reduced to 2D or 1D models. This reduction relies on the assumption that the time evolution does not excite degrees of freedom along the tightly confined directions, where the energy gap is much larger than the interaction energy.

\subsubsection*{Case I: Disk-shaped condensation (3D $\to$ 2D)}
Consider the case where the confinement along the $z$-axis is significantly stronger than in the transverse plane. Physically, the large trap frequency $\omega_z$ creates a large energy gap between the ground state and the excited states along the $z$-axis. Provided the interaction energy is small compared to this gap, the dynamics along the $z$-direction are effectively ``frozen" in the ground state. 

To allow for an exact reduction of the nonlinear terms, we assume the trap frequencies satisfy $m_1 \omega_{z,1} = m_2 \omega_{z,2}$. In dimensionless variables, this implies that the confinement strength $\gamma_z$ dominates over the transverse trapping frequencies (which may vanish), i.e.,
\begin{equation}
    \frac{\gamma_{z,1}}{\sqrt{\varepsilon_1}} = \frac{\gamma_{z,2}}{\sqrt{\varepsilon_2}} \equiv \gamma_z \gg \gamma_{x,j}, \gamma_{y,j} \ge 0, \quad j=1,2.
\end{equation}
Under this condition, the ground states of the two components along the $z$-direction share the same spatial profile. Based on the ``frozen" assumption, we employ the separation of variables ansatz:
\begin{equation}\label{eq: ansatz 2D}
    \Psi(\mathbf{x},t) = \Psi_2(x,y,t) \odot \Psi_{\mathrm{ho}}(z), \quad \Psi_{\mathrm{ho}}(z) = (\psi_{\mathrm{ho}}(z), \psi_{\mathrm{ho}}(z))^T,
\end{equation}
where $\psi_{\mathrm{ho}}(z) = (\frac{\gamma_z}{\pi})^{1/4} \mathrm{e}^{-\frac{\gamma_z z^2}{2}}$. Substituting \eqref{eq: ansatz 2D} into the 3D eGPE, multiplying by $\Psi_{\mathrm{ho}}^{*}(z)$ (element-wise), and integrating over $z$, we first obtain an intermediate equation containing the zero-point energy:
\begin{equation}
    \mathrm{i}\frac{\partial \Psi_2}{\partial t}=-\frac{1}{2}\boldsymbol{\varepsilon}\odot\nabla_\perp^2\Psi_2+(\mathbf{V}_2(x,y)+\mathbf{C})\odot\Psi_2+\mathbf{A}_2(\Psi_2)\odot\Psi_2+\mathbf{B}_2(\Psi_2)\odot\Psi_2,
\end{equation}
where $\nabla_\perp^2 = \partial_x^2 + \partial_y^2$ is the transverse Laplacian, and $\mathbf{C}=(c_1, c_2)^T$ represents the constant energy shift from the $z$-confinement, with 
\[ c_j=\int_{-\infty}^{\infty}\left[\frac{\varepsilon_j}{2}|\psi_{\mathrm{ho}}'(z)|^2+\frac{1}{2}\gamma_{z}^2 z^2|\psi_{\mathrm{ho}}(z)|^2\right]\mathrm{d}z = \frac{\varepsilon_j \gamma_z}{2},\quad j=1,2. \]
Since the physical observables depend only on the modulus of the wave function, we can remove the constant potential term via the gauge transformation $\Psi_2 \to \Psi_2 \odot \mathrm{e}^{-\mathrm{i}\mathbf{C}t}$. This leads to the following 2D eGPE:
\begin{equation}
    \mathrm{i}\frac{\partial \Psi_2}{\partial t}=-\frac{1}{2}\boldsymbol{\varepsilon}\odot\nabla_\perp^2\Psi_2+\mathbf{V}_2(x,y)\odot\Psi_2+\mathbf{A}_2(\Psi_2)\odot\Psi_2+\mathbf{B}_2(\Psi_2)\odot\Psi_2.
\end{equation}
Here, the effective 2D potentials and nonlinear terms are defined as
\begin{align*}
    &\mathbf{V}_2(x,y) = (V_{2,1}, V_{2,2})^T, \quad 
    \mathbf{A}_2(\Psi_2) = (A_{2,1}, A_{2,2})^T, \quad 
    \mathbf{B}_2(\Psi_2) = (B_{2,1}, B_{2,2})^T, \\
    & V_{2,j} = \frac{1}{2}(\gamma_{x,j}^2 x^2 + \gamma_{y,j}^2 y^2), \;\;\; 
    A_{2,j} = \sum_{l=1}^2 \alpha_{2,jl} |\psi_{2,l}|^2, \;\;\;
    B_{2,j} = \delta_2 \beta_j \left(\sum_{l=1}^2 \beta_l |\psi_{2,l}|^2 \right)^{3/2}
\end{align*}
with the renormalized coupling constants calculated analytically as:
\begin{equation*}
    \alpha_{2,jl} = \alpha_{jl} \int_{-\infty}^{\infty} |\psi_{\mathrm{ho}}(z)|^4 \mathrm{d}z = \alpha_{jl}\sqrt{\frac{\gamma_{z}}{2\pi}}, \quad 
    \delta_2 = \delta \int_{-\infty}^{\infty} |\psi_{\mathrm{ho}}(z)|^5 \mathrm{d}z = \delta \sqrt{\frac{2}{5}}\left(\frac{\gamma_z}{\pi}\right)^{3/4}.
\end{equation*}

\subsubsection*{Case II: Cigar-shaped condensation (3D $\to$ 1D)}
Consider the cigar-shaped geometry where the confinement in the transverse plane ($y$-$z$ plane) is much stronger than along the longitudinal $x$-axis. Similar to the disk-shaped case, we assume the trap frequencies satisfy the matching conditions $m_1 \omega_{y,1} = m_2 \omega_{y,2}$ and $m_1 \omega_{z,1} = m_2 \omega_{z,2}$. In dimensionless units, this implies
\begin{equation}
    \frac{\gamma_{y,1}}{\sqrt{\varepsilon_1}} = \frac{\gamma_{y,2}}{\sqrt{\varepsilon_2}} \equiv \gamma_y, \quad \frac{\gamma_{z,1}}{\sqrt{\varepsilon_1}} = \frac{\gamma_{z,2}}{\sqrt{\varepsilon_2}} \equiv \gamma_z, \quad \text{with } \gamma_y, \gamma_z \gg \gamma_{x,j}\ge 0.
\end{equation}
Assuming the transverse degrees of freedom are frozen in the ground state, we use the ansatz:
\begin{equation}
    \Psi(\mathbf{x},t) = \Psi_1(x,t) \odot \Psi_{\mathrm{ho}}(y,z), \quad \Psi_{\mathrm{ho}}(y,z) = (\psi_{\mathrm{ho}}(y,z), \psi_{\mathrm{ho}}(y,z))^T,
\end{equation}
where the normalized transverse ground state is $\psi_{\mathrm{ho}}(y,z) = \frac{(\gamma_y \gamma_z)^{1/4}}{\sqrt{\pi}} \mathrm{e}^{-\frac{\gamma_y y^2 + \gamma_z z^2}{2}}$. Following the same reduction procedure as in the 2D case---integrating out the $y, z$ coordinates and removing the constant energy shift 
via a gauge transformation---we derive the effective 1D eGPE:
\begin{equation}
    \mathrm{i}\frac{\partial \Psi_1}{\partial t}=-\frac{1}{2}\boldsymbol{\varepsilon}\odot\partial_{xx}\Psi_1+\mathbf{V}_1(x)\odot\Psi_1+\mathbf{A}_1(\Psi_1)\odot\Psi_1+\mathbf{B}_1(\Psi_1)\odot\Psi_1.
\end{equation}
Here, $\mathbf{V}_1(x) = (V_{1,1}, V_{1,2})^T$ with $V_{1,j} = \frac{1}{2}\gamma_{x,j}^2 x^2$. The effective interaction terms $\mathbf{A}_1$ and $\mathbf{B}_1$ take the same functional forms as in Eqs. (2.16)-(2.17) but with $\Psi_2$ replaced by $\Psi_1$ and the coefficients $\alpha_{2,jl},\delta_2$ renormalized to
\begin{equation*}
    \alpha_{1,jl} = \alpha_{jl} \int_{\mathbb{R}^2} |\psi_{\mathrm{ho}}|^4 \mathrm{d}y\mathrm{d}z = \alpha_{jl} \frac{\sqrt{\gamma_y\gamma_z}}{2\pi}, \quad 
    \delta_1 = \delta \int_{\mathbb{R}^2} |\psi_{\mathrm{ho}}|^5 \mathrm{d}y\mathrm{d}z = \frac{2\delta}{5} \left(\frac{\gamma_y\gamma_z}{\pi^2}\right)^{3/4}.
\end{equation*}

\subsubsection*{Unified Form}
In summary, the original 3D eGPE and the derived 2D and 1D effective models can be cast into a unified vector form:
\begin{equation}\label{eq: unified eGPE}
    \mathrm{i}\frac{\partial \Psi}{\partial t}=-\frac{1}{2}\boldsymbol{\varepsilon}\odot\nabla^2\Psi+\mathbf{V}_d(\mathbf{x})\odot\Psi+\mathbf{A}_d(\Psi)\odot\Psi+\mathbf{B}_d(\Psi)\odot\Psi, \quad \mathbf{x}\in\mathbb{R}^d,
\end{equation}
where $\nabla^2$ denotes the Laplacian in $d$ dimensions ($d=1,2,3$). For the 3D case ($d=3$), the functions $\mathbf{V}_3$, $\mathbf{A}_3$, and $\mathbf{B}_3$ correspond to the original definitions in Eqs. \eqref{eq: compact eGPE}--\eqref{eq: nonlinear terms 3d}. For the lower-dimensional cases ($d=1,2$), they correspond to the effective terms derived in Case I and Case II with renormalized coupling constants. The normalization condition remains $\int_{\mathbb{R}^d} |\psi_j|^2 \mathrm{d}\mathbf{x} = N_j/N$.

\begin{remark}\label{rem: dimension_reduction}
    We clarify the validity of the dimensional reduction employed here, which depends critically on the comparison between the transverse confinement length $a_{\perp}$ and the characteristic interaction (healing) length $\xi$ \cite{shamriz_suppression_2020}. When $a_{\perp} \ll \xi$, the extreme confinement alters the scattering properties, necessitating a 2D renormalization that leads to logarithmic LHY corrections ($\sim n^2 \ln n$) \cite{lieb2005mathematics,petrov_ultradilute_2016}. In contrast, the condition $a_{\perp} \gg \xi$ implies that the collision dynamics remain effectively three-dimensional. This latter case corresponds to the regime considered in our work, where we retain the 3D LHY nonlinearity ($\sim n^{5/2}$) and perform the reduction solely by integrating out the transverse direction (under the frozen-state assumption) \cite{lieb2005mathematics,shamriz_suppression_2020}. 
\end{remark}

%
%

\subsection{Single-component reduction: The density-locked model}
According to Petrov's theory \cite{petrov_quantum_2015}, a two-component mixture in the droplet regime in free space energetically favors a specific density ratio $n_1/n_2 = \sqrt{a_{22}/a_{11}}$. Physically, if the total particle numbers $N_1$ and $N_2$ deviate from this optimal ratio, the excess atoms are expelled from the droplet core into the background via self-evaporation, leaving a core that satisfies the locking condition (subject to a small critical deviation $\sim \delta a/a_{ii}$). 
Therefore, in the balanced case where the global particle number ratio matches this optimal value (i.e., $N_1/N_2 = \sqrt{a_{22}/a_{11}}$), there are no excess atoms, and the density locking condition $|\psi_1|^2/|\psi_2|^2 = N_1/N_2$ holds globally. Assuming the two components also share the same phase, we can introduce the single-mode ansatz $\psi_j(\mathbf{x},t) = \sqrt{N_j/N}\psi(\mathbf{x},t)$, which reduces the coupled system to a single effective equation \cite{flynn_quantum_2023}. Although derived for free space, this reduction can also be generalized to systems under harmonic confinement \cite{petrov_quantum_2015}.

Substituting this single-mode ansatz into the general energy functional \eqref{eq: energy functional} with $\mathcal{E}_{\mathrm{LHY}}$ given in \eqref{eq:eLHYdensity}, assuming $m_1=m_2=m$ and $V_1=V_2=V$, and utilizing the relation $\delta a = a_{12} + \sqrt{a_{11}a_{22}}$, the effective single-component energy functional is derived as
\begin{align}
    \mathcal{E}(\psi) &= \int_{\mathbb{R}^3} \left[ \frac{\hbar^2}{2m}|\nabla\psi|^2 + V(\mathbf{x})|\psi|^2 + \frac{4\pi\hbar^2}{m}\frac{\delta a \sqrt{a_{11}a_{22}}}{(\sqrt{a_{11}}+\sqrt{a_{22}})^2}|\psi|^4 \right.\\
    &\nonumber\qquad\qquad \left.+ \frac{256\sqrt{\pi}\hbar^2}{15m}(a_{11}a_{22})^{5/4}|\psi|^5 \right] \mathrm{d}\mathbf{x}.
\end{align}
Based on the variation of this functional with respect to $\psi^*$, we can obtain the density-locked eGPE:
\begin{equation}\label{eq: eGPE density locked}
    \mathrm{i}\hbar\frac{\partial \psi}{\partial t} = \left[-\frac{\hbar^2}{2m}\nabla^2 \!+\! V(\mathbf{x}) \!+\! \frac{8\pi\hbar^2}{m}\frac{\delta a \sqrt{a_{11}a_{22}}}{(\sqrt{a_{11}}+\sqrt{a_{22}})^2}|\psi|^2 \!+\! \frac{128\sqrt{\pi}\hbar^2}{3m}(a_{11}a_{22})^{5/4}|\psi|^3 \right]\psi.
\end{equation}
We apply the same non-dimensionalization scaling as in Section~\ref{sec:Non-dim}, i.e., $t \to t/\omega$, $\mathbf{x} \to x_s \mathbf{x}$, and $\psi \to \sqrt{N} x_s^{-3/2} \psi$ with $x_s=\sqrt{\frac{\hbar}{m\omega}}$. The resulting dimensionless equation is
\begin{equation}\label{eq: dimensionless density locked}
    \mathrm{i}\frac{\partial \psi}{\partial t} = \left[-\frac{1}{2}\nabla^2 + V(\mathbf{x}) + \alpha |\psi|^2 + \beta |\psi|^3 \right]\psi, \quad \int_{\mathbb{R}^3} |\psi|^2 \mathrm{d}\mathbf{x} = 1,
\end{equation}
where the dimensionless parameters are
\begin{equation}\label{eq:alpha-and-beta}
    \alpha = \frac{8\pi \delta a \sqrt{a_{11}a_{22}}\,N }{x_s (\sqrt{a_{11}}+\sqrt{a_{22}})^2}, \quad \beta = \frac{128\sqrt{\pi}(a_{11}a_{22})^{5/4}N^{3/2}}{3 x_s^{5/2}}.
\end{equation}
Here, the potential $V(\mathbf{x})$ follows the definitions in the general model.
Similar to the general two-component model, the 3D density-locked equation can be reduced to lower dimensions under strong confinement assumptions. By integrating out the frozen ground state coordinates, the original 3D equation and the effective 1D and 2D models can be expressed in the following unified form:
\begin{equation}\label{eq: unified density locked}
    \mathrm{i}\frac{\partial \psi}{\partial t} = \left[-\frac{1}{2}\nabla^2 + V_d(\mathbf{x}) + \alpha_d |\psi|^2 + \beta_d |\psi|^3 \right]\psi, \quad \mathbf{x} \in \mathbb{R}^d,
\end{equation}
where $\nabla^2$ is the $d$-dimensional Laplacian. The scalar potential $V_d(\mathbf{x})$ is defined as
\begin{equation}
    V_d(\mathbf{x}) = \begin{cases}
        \frac{1}{2}\gamma_x^2 x^2, & d=1, \\
        \frac{1}{2}(\gamma_x^2 x^2 + \gamma_y^2 y^2), & d=2, \\
        \frac{1}{2}(\gamma_x^2 x^2 + \gamma_y^2 y^2 + \gamma_z^2 z^2), & d=3.
    \end{cases}
\end{equation}
The dimension-dependent coefficients $\alpha_d$ and $\beta_d$, obtained after renormalizing the interaction strengths, are given by
\begin{equation}
    \alpha_d = \begin{cases}
        \frac{\sqrt{\gamma_y\gamma_z}}{2\pi} \alpha, & d=1, \\
        \sqrt{\frac{\gamma_z}{2\pi}} \alpha, & d=2, \\
        \alpha, & d=3,
    \end{cases}
    \qquad
    \beta_d = \begin{cases}
          \frac{2}{5}\left(\frac{\gamma_y\gamma_z}{\pi^2}\right)^{3/4} \beta, & d=1, \\
          \sqrt{\frac{2}{5}}\left(\frac{\gamma_z}{\pi}\right)^{3/4} \beta, & d=2, \\
         \beta, & d=3.
    \end{cases}
\end{equation}
The normalization condition for the unified equation is $\int_{\mathbb{R}^d} |\psi|^2 \mathrm{d}\mathbf{x} = 1$.

\begin{remark}
While the density-locked model is formally derived based on ground-state energy minimization, it is also widely employed to describe the dynamics of quantum droplets \cite{flynn_quantum_2023,petrov_quantum_2015}. The validity of this dynamical approximation rests on the assumption of ``low-energy excitations," specifically that the system evolves primarily through in-phase collective modes (e.g., breathing modes) where the density ratio remains locked. If the excitation energy is sufficiently high to trigger out-of-phase motion (spin modes) or significant self-evaporation, the single-component description may break down. A rigorous mathematical analysis of the error bounds for this dynamical reduction is beyond the scope of the current study and remains an interesting topic for future research.
\end{remark}

\section{Numerical methods for ground state solutions}\label{sec:Numerical methods}
In this section, we present efficient numerical methods for computing the ground states of quantum droplets. While our discussion focuses on the general two-component model, the methodology applies directly to the density-locked model as a simplified case.

\subsection{Energy minimization and Euler-Lagrange equations}
The ground state solution $\Phi_g(\mathbf{x})$ is strictly defined as the minimizer of the energy functional $E(\Phi)$ under the normalization constraint:
\begin{equation} \label{eq: minimization problem}
    E(\Phi_g) = \min_{\Phi \in \mathbb{U}} E(\Phi), \;\;\; \mathbb{U} = \left\{ \Phi = (\phi_1, \phi_2)^T  \Big|  E(\Phi)<\infty, \|\phi_j\|^2 = \frac{N_j}{N}, j=1,2 \right\},
\end{equation}
where the dimensionless energy functional $E(\Phi)$ is given by
\begin{equation} \label{eq: energy functional definition}
    E(\Phi) \!=\!\! \int_{\mathbb{R}^d} \!\!\left[ \sum_{j=1}^2\! \left( \frac{\varepsilon_j}{2}|\nabla\phi_j|^2 \!+\! V_{d,j}|\phi_j|^2\! \right) \!+\! \sum_{j,l=1}^2\! \frac{\alpha_{d,jl}}{2}|\phi_j|^2|\phi_l|^2 \!+\! \frac{2\delta_d}{5} \!\left(\sum_{l=1}^2\beta_{d,l}|\phi_l|^2\!\right)^{\!\!5/2} \right]\! \mathrm{d}\mathbf{x}.\!
\end{equation}
Clearly, every critical point (including the minimizer) of the energy minimization problem \eqref{eq: energy functional definition} satisfies the Euler-Lagrange equations:
\begin{equation}\label{eq: nonlinear eigenvalue problem}
    \boldsymbol{\mu}\odot\Phi(\mathbf{x})=-\frac{1}{2}\boldsymbol{\varepsilon}\odot\nabla^2\Phi(\mathbf{x})+\mathbf{V}_d(\mathbf{x})\odot\Phi(\mathbf{x})+\mathbf{A}_d(\Phi)\odot\Phi(\mathbf{x})+\mathbf{B}_d(\Phi)\odot\Phi(\mathbf{x}),
\end{equation}
under the normalization conditions $\|\phi_j\|^2={N_j}/{N},\, j=1,2$.
Here, $\mathbf{V}_d, \mathbf{A}_d, \mathbf{B}_d$ follow the definitions in Section~\ref{sec:eGPE}, and the Lagrange multipliers $\boldsymbol{\mu}=(\mu_1, \mu_2)^T$ are also called the chemical potentials. When $\Phi$ is an eigenfunction, the corresponding chemical potentials $\mu_j$ ($j=1,2$) can be explicitly computed as
\begin{align} \label{eq: chemical potential def}
    \mu_j (\Phi) &= \frac{N}{N_j} \int_{\mathbb{R}^d} \left[ \frac{\varepsilon_j}{2}|\nabla\phi_j|^2 + V_{d,j}|\phi_j|^2 + A_{d,j}(\Phi)|\phi_j|^2 + B_{d,j}(\Phi)|\phi_j|^2 \right] \mathrm{d}\mathbf{x}.
\end{align}

\begin{remark}
    We note that the nonlinear eigenvalue problem \eqref{eq: nonlinear eigenvalue problem} constitutes the time-independent eGPE, which can be recovered from the time-dependent eGPE \eqref{eq: unified eGPE} under the stationary wave function ansatz $\Psi(\mathbf{x},t) = e^{-\mathrm{i}\boldsymbol{\mu} t}\odot\Phi(\mathbf{x})$.
\end{remark}

\begin{remark}
    It is worth emphasizing that for nonlinear systems, the total energy $E(\Phi)$ is generally not equal to the sum of the chemical potentials (eigenvalues). Unlike linear problems, the ground state—defined as the global minimizer of the energy functional—does not necessarily correspond to the eigenstate associated with the smallest chemical potentials. Therefore, determining the ground state should be treated as a constrained optimization problem rather than merely solving for the smallest eigenvalues.
\end{remark}

For the conventional two-component BEC, i.e., no $\mathbf{B}(\Phi)$ term, the minimizer was computed by the CNGF or the imaginary time method, etc. Here we extend the CNGF and its discretization to quantum droplets in Bose-Bose mixtures.

\subsection{Normalized gradient flows}
We construct the following CNGF for the ground state computation:
\begin{equation}\label{eq: CNGF}
    \!\!\left\{\begin{aligned}
        &\partial_t\Phi(\mathbf{x}, t)=\frac{1}{2}\boldsymbol{\varepsilon}\odot\nabla^2\Phi-\big[\mathbf{V}_d(\mathbf{x})+\mathbf{A}_d(\Phi)+\mathbf{B}_d(\Phi)-\boldsymbol{\mu}({\Phi}(\cdot,t))\big]\odot\Phi,\;\; t\ge0,\\
        &\Phi(\mathbf{x},0)=\Phi_0(\mathbf{x}), \quad \text{with} \quad \|\phi_{0,j}\|^2 = \frac{N_j}{N}, \quad j=1,2,
    \end{aligned}\right.
\end{equation}
where $\boldsymbol{\mu}({\Phi}(\cdot,t))$ ensures the normalization conservation and is defined by Eq. \eqref{eq: chemical potential def}.
The CNGF satisfies the following properties.

\begin{theorem}
    Suppose $V_{d,j}(\mathbf{x}) \ge 0$ and the initial data $\Phi_0$ satisfies the normalization condition. Then the CNGF \eqref{eq: CNGF} preserves the normalization and diminishes the energy, i.e.,
    \begin{align}
            &\|\phi_j(\cdot,t)\|^2=\frac{N_j}{N},\quad t\ge0,\quad j=1,2,\label{eq: normalization conservation of CNGF}\\
            &\frac{\mathrm{d}}{\mathrm{d}t}E(\Phi(\cdot,t)) = -2 \sum_{j=1}^2 \|\partial_t \phi_j(\cdot,t)\|^2 \le 0,
    \end{align}
    which implies $E(\Phi(\cdot,t_1))\ge E(\Phi(\cdot,t_2))$ for any $0\le t_1\le t_2<\infty$.
\end{theorem}

\begin{proof}
    First, we verify the normalization conservation. Recall that for a complex function, $\frac{\mathrm{d}}{\mathrm{d}t}\|\phi_j\|^2 = \int (\partial_t \phi_j \phi_j^* + \phi_j \partial_t \phi_j^*) \mathrm{d}\mathbf{x} = 2 \mathrm{Re} \int \phi_j^* \partial_t \phi_j \mathrm{d}\mathbf{x}$. Multiplying the $j$-th component of the CNGF Eq. \eqref{eq: CNGF} by $\phi_j^*$, integrating over $\mathbb{R}^d$, we obtain
    \begin{equation*}
        \int_{\mathbb{R}^d} \phi_j^* \partial_t \phi_j \mathrm{d}\mathbf{x} \!=\! \int_{\mathbb{R}^d}\! \phi_j^* \left[ \frac{\varepsilon_j}{2}\nabla^2\phi_j - V_{d,j}\phi_j - A_{d,j}(\Phi)\phi_j - B_{d,j}(\Phi)\phi_j + \mu_{j}(\Phi(\cdot,t))\phi_j \right] \!\mathrm{d}\mathbf{x}.
    \end{equation*}
    Let $\mathcal{H}_j$ denote the nonlinear Hamiltonian operator inside the brackets (excluding $\mu_{j}(\Phi)$). Since $\mathcal{H}_j$ is Hermitian (i.e., $\int \phi^* \mathcal{H} \phi$ is real) and $\mu_{j}(\Phi)$ is real by definition, the integral on the right-hand side is purely real. 
    By the definition of $\mu_{j}(\Phi)$ in Eq. \eqref{eq: chemical potential def}, we have $\mu_{j}(\Phi) = \frac{\int \phi_j^* \mathcal{H}_j \phi_j}{\|\phi_j\|^2}$, which ensures that the right-hand side of the integrated equation is exactly zero. Thus, $\frac{\mathrm{d}}{\mathrm{d}t}\|\phi_j\|^2 = 2 \mathrm{Re}(0) = 0$, confirming Eq. \eqref{eq: normalization conservation of CNGF}.

    Next, we compute the time derivative of the energy functional $E(\Phi)$. Using the chain rule for complex functionals, we have
    \begin{equation*}
        \frac{\mathrm{d}}{\mathrm{d}t}E(\Phi) = \sum_{j=1}^2 \int_{\mathbb{R}^d} \left( \frac{\delta E}{\delta \phi_j} \partial_t \phi_j + \frac{\delta E}{\delta \phi_j^*} \partial_t \phi_j^* \right) \mathrm{d}\mathbf{x} = 2 \mathrm{Re} \sum_{j=1}^2 \int_{\mathbb{R}^d} \frac{\delta E}{\delta \phi_j^*} \partial_t \phi_j^* \mathrm{d}\mathbf{x}.
    \end{equation*}
    Computing the variational derivative with respect to $\phi_j^*$ yields
    \begin{equation*}
        \frac{\delta E}{\delta \phi_j^*} = -\frac{\varepsilon_j}{2} \nabla^2 \phi_j + V_{d,j}\phi_j + A_{d,j}(\Phi)\phi_j + B_{d,j}(\Phi)\phi_j,\quad j=1,2.
    \end{equation*}
    Comparing this with the CNGF equation \eqref{eq: CNGF}, we observe that $\partial_t \phi_j = - \frac{\delta E}{\delta \phi_j^*} + \mu_{j}(\Phi)\phi_j$. Or equivalently, $\frac{\delta E}{\delta \phi_j^*} = \mu_{j}(\Phi)\phi_j - \partial_t \phi_j$. Substituting this back into the energy derivative gives
    \begin{align*}
        \frac{\mathrm{d}}{\mathrm{d}t}E(\Phi) &= 2 \mathrm{Re} \sum_{j=1}^2 \int_{\mathbb{R}^d} (\mu_{j}(\Phi)\phi_j - \partial_t \phi_j) \partial_t \phi_j^* \mathrm{d}\mathbf{x} \\
        &= \sum_{j=1}^2 \left[ 2\mu_{\Phi,j} \,\mathrm{Re} \int_{\mathbb{R}^d} \phi_j \partial_t \phi_j^* \mathrm{d}\mathbf{x} - 2 \int_{\mathbb{R}^d} |\partial_t \phi_j|^2 \mathrm{d}\mathbf{x} \right].
    \end{align*}
    From the normalization conservation, we know $\frac{\mathrm{d}}{\mathrm{d}t}\|\phi_j\|^2 = 2 \mathrm{Re} \int \phi_j \partial_t \phi_j^* \mathrm{d}\mathbf{x} = 0$. Therefore, the first term vanishes, and we arrive at the energy diminishing property:
    \[ \frac{\mathrm{d}}{\mathrm{d}t}E(\Phi) = -2 \sum_{j=1}^2 \|\partial_t \phi_j\|^2 \le 0. \]
\end{proof}

The above energy diminishing property provides a computational stability of CNGF for computing the ground states.



To facilitate the implementation of the CNGF, we consider two gradient flow computational models based on discrete normalization projection. 
Let $\tau = \Delta t > 0$ be the time step size and $t_n = n \tau$ for $n \ge 0$. 

\subsubsection*{Gradient Flow with Discrete Normalization (GFDN)}
The GFDN method consists of two steps: first, it evolves the wave function by neglecting the normalization constraint (steepest descent of the unconstrained energy), and second, it projects the solution back onto the normalization manifold. The formulation reads
\begin{equation}\label{eq: GFDN}
    \left\{\begin{aligned}
        &\partial_t\Phi(\mathbf{x}, t)=\frac{1}{2}\boldsymbol{\varepsilon}\odot\nabla^2\Phi-\big[\mathbf{V}_d(\mathbf{x})+\mathbf{A}_d(\Phi)+\mathbf{B}_d(\Phi)\big]\odot\Phi, \quad t \in [t_n, t_{n+1}),\\
        &\Phi(\mathbf{x}, t_{n+1}) \triangleq \Phi(\mathbf{x}, t_{n+1}^+) \!=\! \left(\!\! \sqrt{\!\frac{N_1}{N}}\frac{\phi_1(\mathbf{x}, t_{n+1}^-)}{\|\phi_1(\cdot, t_{n+1}^-)\|}, \sqrt{\!\frac{N_2}{N}}\frac{\phi_2(\mathbf{x}, t_{n+1}^-)}{\|\phi_2(\cdot, t_{n+1}^-)\|} \!\right)^T\!\!, \; n\geq0,\\
        &\Phi(\mathbf{x},0)=\Phi_0(\mathbf{x}), \quad \mathbf{x}\in\mathbb{R}^d,
    \end{aligned}\right.
\end{equation}
where $\Phi(\mathbf{x}, t_{n+1}^-)$ is the solution of the first equation at $t = t_{n+1}$ obtained from the initial value $\Phi(\mathbf{x}, t_n)$.

For the linear case (i.e., without interaction terms), the GFDN ensures energy diminishing unconditionally:
\begin{theorem}[\cite{bao2004computing}]
    Suppose $V_{d,j}(\mathbf{x})\ge 0$ and the interactions vanish ($\alpha_{jl}=\beta_{j}=0$). The GFDN scheme \eqref{eq: GFDN} is normalization conserving and energy diminishing for any time step $\tau>0$ and initial data $\Phi_0$, i.e.,
    \begin{equation}
        E\left(\Phi(\cdot,t_{n+1})\right) \le E\left(\Phi(\cdot,t_{n})\right) \le \cdots \le E\left(\Phi_0\right), \quad n \ge 0.
    \end{equation}
\end{theorem}

\subsubsection*{Gradient Flow with Lagrange Multiplier (GFLM)}
Despite the broad effectiveness of GFDN, standard temporal discretizations (such as the backward-forward Euler scheme) under the GFDN framework may introduce splitting errors that prevent convergence to the exact ground state \cite{faou2018convergence,liu2021normalized}. Following \cite{liu2021normalized}, we propose the GFLM method for the quantum droplet system by incorporating explicit Lagrange multiplier terms into the gradient flow evolution, reading as
\begin{equation}\label{eq: GFLM}
    \left\{\begin{aligned}
        &\partial_t\Phi = \frac{1}{2}\boldsymbol{\varepsilon}\odot\nabla^2\Phi-\big[\mathbf{V}_d(\mathbf{x})+\mathbf{A}_d(\Phi)+\mathbf{B}_d(\Phi)-\boldsymbol{\mu}({\Phi}(\cdot,t_n))\big]\odot\Phi,\;\;\; t \in [t_n, t_{n+1}),\\
        &\Phi(\mathbf{x}, t_{n+1}) \triangleq \Phi(\mathbf{x}, t_{n+1}^+) = \left(\! \sqrt{\frac{N_1}{N}}\frac{\phi_1(\mathbf{x}, t_{n+1}^-)}{\|\phi_1(\cdot, t_{n+1}^-)\|}, \sqrt{\frac{N_2}{N}}\frac{\phi_2(\mathbf{x}, t_{n+1}^-)}{\|\phi_2(\cdot, t_{n+1}^-)\|} \!\right)^T\!\!, \; n\geq0,\\   &\Phi(\mathbf{x},0)=\Phi_0(\mathbf{x}), \quad \mathbf{x}\in\mathbb{R}^d,
    \end{aligned}\right.
\end{equation}
where the chemical potential $\boldsymbol{\mu}({\Phi}(\cdot,t_n))$ is fixed at the beginning of each time interval using the formula in Eq. \eqref{eq: chemical potential def}. 
The key difference between GFDN and GFLM is the inclusion of the explicit Lagrange multiplier term $\boldsymbol{\mu}({\Phi}(\cdot,t_n))\odot\Phi$, which helps to correct the stationary state of the discrete scheme.

\subsection{Full discretization: BESP and BFSP schemes}\label{sec:Full discretization}
In this subsection, we present the full discretization of the GFDN and GFLM methods using the linearized Backward Euler Sine-Pseudospectral (BESP) and Backward-Forward Sine-Pseudospectral (BFSP) schemes. For simplicity, we illustrate the method for the 1D case ($d=1$) on a bounded domain $[a, b]$ with homogeneous Dirichlet boundary conditions. Generalizations to 2D and 3D are straightforward via tensor products.

Let $M$ be an even positive integer, and define the spatial mesh size $h = (b-a)/M$ and time step size $\tau = \Delta t > 0$. The grid points and time steps are given by $x_j = a + j h$ ($j=0,1, \ldots, M$) and $t_n = n \tau$ ($n=0,1, \ldots$), respectively. Let $\Phi_j^n$ denote the numerical approximation of $\Phi(x_j, t_n)$. We use the sine-spectral method for spatial discretization. The spectral second-order differential operator $D_{xx}^s$ acting on a vector $U = (U_0, \dots, U_M)^T$ (with $U_0=U_M=0$) is defined as
\begin{equation}
    (D_{xx}^s U)_j = -\frac{2}{M}\sum_{l=1}^{M-1} \mu_l^2 (\hat{U})_l \sin(\mu_l(x_j-a)), \quad j=1, \ldots, M-1,
\end{equation}
where $\mu_l = \frac{\pi l}{b-a}$, and $(\hat{U})_l$ are the sine transform coefficients:
\begin{equation}
    (\hat{U})_l = \sum_{j=1}^{M-1} U_j \sin(\mu_l(x_j-a)), \quad l=1, \ldots, M-1.
\end{equation}

\subsubsection*{BESP Scheme for GFLM}
The linearly implicit BESP scheme for GFLM is given by:
\begin{equation}\label{eq: BESP}
    \frac{\Phi_j^{*} - \Phi_j^n}{\tau} = \frac{1}{2}\boldsymbol{\varepsilon}\odot (D_{xx}^s\Phi^{*})_j - \big[\mathbf{V}_1(x_j) + \mathbf{A}_1(\Phi_j^n) + \mathbf{B}_1(\Phi_j^n)\big]\odot\Phi_j^{*} + \boldsymbol{\mu}({\Phi}^n)\odot\Phi_j^n,
\end{equation}
for $j=1, \ldots, M-1$, with boundary conditions $\Phi_0^* = \Phi_M^* = 0$. The intermediate solution $\Phi^*$ is then projected to satisfy the normalization:
\begin{equation}
    (\phi_{l})_j^{n+1} = \sqrt{\frac{N_l}{N}} \frac{(\phi_{l}^*)_j}{\|\phi_{l}^*\|_h}, \quad l=1,2,
\end{equation}
where $\|\cdot\|_h$ denotes the discrete $L^2$ norm.
Eq. \cref{eq: BESP} constitutes a linear system for $\Phi^*$ with spatially variable coefficients. Solving this variable-coefficient linear system directly is computationally expensive. Therefore, we solve it iteratively using a stabilization term with the parameter $\boldsymbol{\gamma}=(\gamma_1, \gamma_2)^T$:
\begin{align}\label{eq: iterative solver}
    \frac{\Phi_j^{*,m+1} - \Phi_j^n}{\tau} &= \frac{1}{2}\boldsymbol{\varepsilon}\odot (D_{xx}^s\Phi^{*,m+1})_j - \boldsymbol{\gamma}\odot\Phi_j^{*,m+1} \nonumber \\
    &\quad + \left[\boldsymbol{\gamma} - \mathbf{V}_1(x_j) - \mathbf{A}_1(\Phi_j^n) - \mathbf{B}_1(\Phi_j^n)\right]\odot\Phi_j^{*,m} + \boldsymbol{\mu}({\Phi}^n)\odot\Phi_j^n,
\end{align}
with $\Phi_j^{*,0} = \Phi_j^n$. This linear system with constant coefficients can be solved efficiently using the fast Fourier transform (FFT).

\subsubsection*{BFSP Scheme for GFLM}
Taking only one iteration step ($m=0$) in Eq. \eqref{eq: iterative solver} yields the BFSP scheme:
\begin{align}\label{eq: BFSP}
    \frac{\Phi_j^{*} - \Phi_j^n}{\tau} &= \frac{1}{2}\boldsymbol{\varepsilon}\odot (D_{xx}^s\Phi^{*})_j - \boldsymbol{\gamma}\odot\Phi_j^{*} \nonumber \\
    &\quad + \left[\boldsymbol{\gamma} - \mathbf{V}_1(x_j) - \mathbf{A}_1(\Phi_j^n) - \mathbf{B}_1(\Phi_j^n)\right]\odot\Phi_j^{n} + \boldsymbol{\mu}({\Phi}^n)\odot\Phi_j^n.
\end{align}
The projection step remains the same. The schemes for GFDN are obtained by simply removing the Lagrange multiplier term $\boldsymbol{\mu}({\Phi}^n)\odot\Phi_j^n$.

\begin{remark}
    Due to the competing nonlinear interactions in quantum droplets, determining an optimal stabilization parameter $\boldsymbol{\gamma}$ is difficult, with no theoretical guarantee for unconditional convergence for an arbitrary time step $\tau$. We adapt the heuristic strategy used in standard BECs \cite{bao2006efficient}:
    \begin{equation}
         \gamma_l=\frac{b_{\max}^l+b_{\min}^l}{2}, \quad l=1,2,
    \end{equation}
    where $b_{\max}^l$ and $b_{\min}^l$ denote the maximum and minimum values of the effective potential $V_{l}+A_{l}(\Phi^n)+B_{l}(\Phi^n)$, respectively. In practice, this choice effectively balances stability and efficiency.
\end{remark}

\begin{remark}
For radially symmetric potentials $V(\mathbf{x}) = V(|\mathbf{x}|)$ (including the free-space case), the ground state preserves radial symmetry. Consequently, the original 3D problem reduces to a 1D equation on the semi-infinite interval $[0, \infty)$. The numerical solution for this reduced system follows the standard radial discretization scheme (see, e.g., Ref.~\cite{bao2004computing}), and we thus omit the implementation details here.
\end{remark}

\subsection{Choice of initial data}
The gradient flow methods (GFDN and GFLM) are iterative processes that require an initial guess $\Phi_0(\mathbf{x})$ to start. A well-chosen initial guess is crucial for two reasons: it significantly accelerates the convergence rate, and it helps the algorithm avoid being trapped in local energy minimizers (excited states), thereby ensuring convergence to the true ground state. We provide effective initial guesses for different physical regimes as follows.

\subsubsection*{Weak-coupling regime}
In the weak-coupling regime, the ground state profile is dominated by the kinetic energy or external potential. We typically adopt a Gaussian profile as the initial guess.

\begin{itemize}
	\item \textbf{Harmonic potential.} 
When the system is confined by a harmonic trap, we choose the ground state of the corresponding non-interacting linear harmonic oscillator:
\begin{equation}
    \phi_j^{(0)}(\mathbf{x}) = \prod_{\nu} \left( \frac{\gamma_{\nu,j}}{\pi\sqrt{\varepsilon_j}} \right)^{1/4} \exp\left( -\frac{\gamma_{\nu,j} x_\nu^2}{2\sqrt{\varepsilon_j}} \right), \quad j=1,2,
\end{equation}
where the product runs over the spatial dimensions $\nu \in \{x,y,z\}$.

	\item \textbf{Free space.} 
In the absence of external potential ($V(\mathbf{x})=0$), we employ a Gaussian ansatz with a tunable width $\sigma_j$:
\begin{equation}
    \phi_j^{(0)}(\mathbf{x}) = \frac{1}{(\pi \sigma_j^2)^{d/4}} \exp\left( -\frac{|\mathbf{x}|^2}{2\sigma_j^2} \right), \quad j=1,2.
\end{equation}
The width parameters $\sigma_j > 0$ are chosen empirically (e.g., $\sigma_j \in [0.1, 5]$) or determined by a continuation technique from solutions with larger particle numbers.
\end{itemize}

\subsubsection*{Strong-coupling regime}
In the strong-coupling regime in free space ($V(\mathbf{x})=0$), we adopt the TFA. By neglecting the kinetic energy term, the ground state reduces to a uniform liquid drop with constant densities confined within a finite radius $R_d$:
\begin{equation}
    \phi_j^{s}(\mathbf{x})=
    \begin{cases}
        \sqrt{\rho_j}, & |\mathbf{x}| \le R_d,\\
        0, & |\mathbf{x}| > R_d,
    \end{cases} \quad j=1,2.
\end{equation}
The equilibrium densities $\rho_1$ and $\rho_2$ are determined by minimizing the bulk energy density:
\begin{equation}
    \sqrt{\rho_1} = -\frac{5}{6}\frac{\alpha_{11}N_1^2+2\alpha_{12}N_1N_2+\alpha_{22}N_2^2}{\delta\left(\beta_1 N_1^{4/5}+\beta_2 N_2N_1^{-1/5}\right)^{5/2}}, \quad
    \sqrt{\rho_2} = -\frac{5}{6}\frac{\alpha_{11}N_1^2+2\alpha_{12}N_1N_2+\alpha_{22}N_2^2}{\delta\left(\beta_1 N_1N_2^{-1/5}+\beta_2 N_2^{4/5}\right)^{5/2}}.
\end{equation}
Consequently, the droplet radius $R_d$ is uniquely determined by the normalization condition $\|\phi_j^s\|^2 = N_j/N$. This Thomas-Fermi (TF) profile serves as a robust initial guess for simulations with large particle numbers.

For systems with a non-vanishing potential or partial confinement, a continuation technique is often more efficient.


\begin{remark}
    The validity of the TF ansatz relies on the existence of a self-bound state, which requires the mean-field attraction to balance the LHY repulsion. This imposes a strict constraint on the particle number ratio:
    \begin{equation}
        \sum_{j,l=1}^2 \alpha_{jl}N_j N_l < 0 \implies \frac{N_1}{N_2} \in \left(\frac{-\alpha_{12}-\sqrt{D}}{\alpha_{11}}, \frac{-\alpha_{12}+\sqrt{D}}{\alpha_{11}}\right),
    \end{equation}
    where $D=\alpha_{12}^2-\alpha_{11}\alpha_{22} > 0$. Outside this range, the mixture cannot form a stable droplet, rendering the single-radius TF ansatz invalid.
\end{remark}

\subsection{Application to the density-locked model}
Analogous to the general two-component case, the ground state solution $\phi_g(\mathbf{x})$ of the density-locked model is defined as the global minimizer of the energy functional:
\begin{equation}\label{eq: ground state as the minimization of energy functional}
    E(\phi_g) = \min_{\phi \in \mathbb{S}} E(\phi), \quad \mathbb{S} = \left\{ \phi \ \Big| \ \|\phi\|^2 = 1, \ E(\phi) < \infty \right\},
\end{equation}
where the dimensionless energy functional is
\begin{equation} \label{eq: energy functional density locked}
    E(\phi) = \int_{\mathbb{R}^d} \left[ \frac{1}{2}|\nabla\phi|^2 + V_d(\mathbf{x})|\phi|^2 + \frac{\alpha_d}{2}|\phi|^4 + \frac{2\beta_d}{5}|\phi|^5 \right] \mathrm{d}\mathbf{x}.
\end{equation}
Every minimizer must satisfy the corresponding Euler-Lagrange equation, which reads as the nonlinear eigenvalue problem:
\begin{equation}\label{eq: nonlinear eigenvalue problem density locked}
    \mu\phi(\mathbf{x}) = -\frac{1}{2}\nabla^2\phi(\mathbf{x}) + V_d(\mathbf{x})\phi(\mathbf{x}) + \alpha_d|\phi(\mathbf{x})|^2\phi(\mathbf{x}) + \beta_d|\phi(\mathbf{x})|^3\phi(\mathbf{x}),
\end{equation}
under the normalization condition $\|\phi\|^2 = 1$, and the corresponding chemical potential $\mu$ is given by
\begin{equation}
    \mu = \int_{\mathbb{R}^d} \left[ \frac{1}{2}|\nabla\phi|^2 + V_d(\mathbf{x})|\phi|^2 + \alpha_d|\phi|^4 + \beta_d|\phi|^5 \right] \mathrm{d}\mathbf{x}.
\end{equation}
The GFDN and GFLM methods introduced in Sections 3.2--3.3 can be applied directly to solve this problem by reducing the vector fields to a scalar field.

Regarding the existence of ground states, we have the following result established in \cite{Luo2021DCDSB} (recall that $0>\alpha\propto N,0<\beta\propto N^{3/2}$ according to \eqref{eq:alpha-and-beta}): 
\begin{theorem}[Existence of ground state \cite{Luo2021DCDSB}]\label{Thm: Existence of ground state} For density-locked model in three dimension, the following hold:
\begin{enumerate}[(1)]
    \item  If $V(\mathbf{x})\ge0\; (\mathbf{x}\in\mathbb{R}^3)$ satisfies $\lim_{|\mathbf{x}|\to\infty}V(\mathbf{x})=\infty$, then there exists ground-state solution for \eqref{eq: ground state as the minimization of energy functional} for any atomic number $N$.
\item If $V(\mathbf{x})=0$, there is a critical atomic number $N_c$. When $N<N_c$, there is no ground state for \eqref{eq: ground state as the minimization of energy functional}, and $E(\phi)>0$ for all $\phi\in\mathbb{S}$. When $N>N_c$, there is a ground state for \eqref{eq: ground state as the minimization of energy functional}. Moreover, the ground state is (up to translation) radially symmetric and there is some real number $\theta\in\mathbb{R}$ such that $\phi=e^{\mathrm{i}\theta}|\phi|$ and $|\phi(\mathbf{x})|>0$ for all $\mathbf{x}\in\mathbb{R}^3$. 
\end{enumerate}
\end{theorem}

\begin{remark}
    The critical particle number $N_c$ represents the threshold for the existence of a self-bound ground state in free space. 
    It is important to note that analytical estimates for $N_c$ (e.g., in \cite{petrov_quantum_2015}) typically rely on a Gaussian variational ansatz (assuming a soliton-like profile). However, such approximations may deviate from the exact value as they cannot fully capture the deformation of the wave function near the transition. We will determine the precise value of $N_c$ numerically in Section~\ref{sec:Numerical results}.
\end{remark}

\noindent\textbf{Choice of initial data.} 
Similar to the full two-component model, a suitable initial guess is essential for efficient computation. We construct $\phi^{(0)}(\mathbf{x})$ based on the interaction regime:

(i) \textit{Weak-coupling regime.} We adopt a Gaussian profile. For harmonic confinement, we use the ground state of the linear oscillator $\phi^{(0)}(\mathbf{x}) = \prod_{\nu} \left(\frac{\gamma_\nu}{\pi}\right)^{1/4} \mathrm{e}^{-\frac{\gamma_\nu x_\nu^2}{2}}$. In free space ($V(\mathbf{x})=0$), we employ a Gaussian ansatz with a tunable width $\sigma$: $\phi^{(0)}(\mathbf{x}) = (\pi\sigma^2)^{-d/4} \mathrm{e}^{-|\mathbf{x}|^2/(2\sigma^2)}$.

(ii) \textit{Strong-coupling regime.} In free space, the ground state exhibits a flat-top profile. We adopt the TF ansatz:
\begin{equation}
    \phi^{s}(\mathbf{x}) = 
    \begin{cases}
        \sqrt{\rho_s}, & |\mathbf{x}| \le R_d, \\
        0, & |\mathbf{x}| > R_d.
    \end{cases}
\end{equation}
The constant amplitude is explicitly given by $\sqrt{\rho_s} = -{5\alpha}/{6\beta}$ (valid for $\alpha < 0, \beta > 0$), corresponding to the equilibrium density of the uniform system. The droplet radius $R_d$ is determined by the normalization $\|\phi^s\|^2 = 1$. For cases with confinement, the continuation technique described in Section~\ref{sec:Full discretization} is recommended.

\section{Numerical results}\label{sec:Numerical results}
In this section, we present numerical results to validate the efficiency of the proposed methods and to investigate the ground state properties of quantum droplets. We use the physical parameters of a ${}^{39}\mathrm{K}$ mixture \cite{ferioli_dynamical_2020,semeghini_self-bound_2018}: the scattering lengths are $a_{11}=69.99a_0$, $a_{12}=-53.37a_0$, and $a_{22}=34.11a_0$, where $a_0$ is the Bohr radius. 

\subsection{Comparison of numerical methods}
We first compare the performance of different discretization schemes: GFDN-BESP, GFDN-BFSP, GFLM-BESP, and GFLM-BFSP. We consider both the density-locked model and the general two-component model in $3$D.

\textbf{Example 4.1} (Density-locked model). 
We consider the density-locked model in $3$D with $V(x)=0$ and $N=10^4$. The computational domain is $[-32,32]$ with a mesh size $h=1/32$. The scaling frequency is $\omega=50$, resulting in the dimensionless parameters $\alpha=-2.5474$ and $\beta=0.3315$ according to \eqref{eq:alpha-and-beta}. The steady state is considered reached when the residual satisfies $\max_{j}|\phi_j^{n+1}-\phi_j^n|/\tau \le 10^{-12}$.
Table \ref{tab: ex4.1} lists the ground state energy $E_g$, chemical potential $\mu_g$, CPU time, and the maximum residual for different time steps $\tau$.

\begin{table}[htbp]
    \centering
    \caption{Comparison of different methods for the density-locked model (Example 4.1).}
    \label{tab: ex4.1}
    \small
    \begin{tabular}{clcccc}
    \hline
        Method & $\tau$ & CPU(s) & $E_g$ & $\mu_g$ & Residual \\ \hline
        ~ & 1 & \multicolumn{4}{c}{Diverged} \\
        ~ & 0.5 & 0.53 & -2.338e-01 & -6.678e-01 & 1.881e-12 \\ 
        GFDN-BESP & 0.1 & 0.90 & -2.338e-01 & -6.678e-01 & 1.413e-12 \\ 
        ~ & 0.05 & 1.66 & -2.338e-01 & -6.678e-01 & 1.264e-12 \\ 
        ~ & 0.01 & 6.36 & -2.338e-01 & -6.678e-01 & 1.145e-12 \\ 
        ~ & 0.005 & 12.88 & -2.338e-01 & -6.678e-01 & 1.244e-12 \\ \hline
        ~ & 0.5 & 0.09 & -2.202e-01 & -5.536e-01 & 5.440e-02 \\ 
        ~ & 0.1 & 0.31 & -2.330e-01 & -6.426e-01 & 1.835e-02 \\ 
        ~ & 0.05 & 0.58 & -2.336e-01 & -6.550e-01 & 9.871e-03 \\ 
        GFDN-BFSP & 0.01 & 2.58 & -2.338e-01 & -6.652e-01 & 2.096e-03 \\ 
        ~ & 0.005 & 5.14 & -2.338e-01 & -6.665e-01 & 1.056e-03 \\  \hline
        ~ & 0.5 & 0.37 & -2.338e-01 & -6.678e-01 & 1.766e-12 \\ 
        ~ & 0.1 & 0.66 & -2.338e-01 & -6.678e-01 & 1.227e-12 \\ 
        GFLM-BESP & 0.05 & 1.15 & -2.338e-01 & -6.678e-01 & 1.164e-12 \\ 
        ~ & 0.01 & 4.69 & -2.338e-01 & -6.678e-01 & 1.285e-12 \\ 
        ~ & 0.005 & 8.98 & -2.338e-01 & -6.678e-01 & 1.204e-12 \\ \hline
        ~ & 0.5 & 0.08 & -2.338e-01 & -6.678e-01 & 1.665e-12 \\ 
        ~ & 0.1 & 0.28 & -2.338e-01 & -6.678e-01 & 1.254e-12 \\ 
        ~ & 0.05 & 0.56 & -2.338e-01 & -6.678e-01 & 1.190e-12 \\ 
        GFLM-BFSP & 0.01 & 2.60 & -2.338e-01 & -6.678e-01 & 1.258e-12 \\ 
        ~ & 0.005 & 5.25 & -2.338e-01 & -6.678e-01 & 1.228e-12 \\ \hline
    \end{tabular}
\end{table}

\textbf{Example 4.2} (General two-component model). 
We utilize the same physical parameters but solve the general two-component eGPE. We choose $V(x)=0$ and $N=4\times10^3$, with the particle number ratio fixed at the optimal value $N_1/N_2=\sqrt{a_{22}/a_{11}}$. The dimensionless coefficients are computed as $\alpha_{11}=32.6631$, $\alpha_{12}=-24.9068$, $\alpha_{22}=15.9185$, $\beta_1=6.4981\times 10^{-4}$, $\beta_2=3.1669\times 10^{-4}$, and $\delta=1.9132\times 10^{7}$. The convergence criterion is $\frac{1}{\tau} \sum_{j=1}^2 \max_{l} |(\phi_j)_{l}^{n+1}-(\phi_j)_l^n| \le 10^{-12}$.
Table \ref{tab: ex4.2} presents the comparison results.

\begin{table}[htbp]
    \centering
    \caption{Comparison of different methods for the general two-component model (Example 4.2).}
    \label{tab: ex4.2}
    \small
    \begin{tabular}{clccccc}
    \hline
        Method & $\tau$ & CPU(s) & $E_g$ & $\mu_{g,1}$ & $\mu_{g,2}$ & Residual \\ \hline
        ~ & 0.5 & \multicolumn{5}{c}{Diverged} \\
        ~ & 0.1 & 6.71 & -4.085e-02 & -1.375e-01 & -1.081e-01 & 1.305e-12 \\ 
        GFDN-BESP & 0.05 & 11.24 & -4.085e-02 & -1.375e-01 & -1.081e-01 & 1.336e-12 \\ 
        ~ & 0.01 & 50.69 & -4.085e-02 & -1.375e-01 & -1.081e-01 & 1.374e-12 \\ 
        ~ & 0.005 & 91.81 & -4.085e-02 & -1.375e-01 & -1.081e-01 & 1.408e-12 \\ \hline
        ~ & 0.1 & 2.62 & -4.085e-02 & -1.366e-01 & -1.072e-01 & 6.445e-04 \\ 
        ~ & 0.05 & 4.62 & -4.085e-02 & -1.370e-01 & -1.076e-01 & 3.269e-04 \\ 
        GFDN-BFSP & 0.01 & 22.19 & -4.085e-02 & -1.374e-01 & -1.080e-01 & 6.613e-05 \\ 
        ~ & 0.005 & 44.30 & -4.085e-02 & -1.375e-01 & -1.080e-01 & 3.311e-05 \\ 
        ~ & 0.001 & 215.70 & -4.085e-02 & -1.375e-01 & -1.080e-01 & 6.630e-06 \\ \hline
        ~ & 0.5 & \multicolumn{4}{c}{Diverged} \\
        ~ & 0.1 & 5.77 & -4.085e-02 & -1.375e-01 & -1.081e-01 & 1.226e-12 \\ 
        GFLM-BESP & 0.05 & 9.83 & -4.085e-02 & -1.375e-01 & -1.081e-01 & 1.328e-12 \\ 
        ~ & 0.01 & 43.70 & -4.085e-02 & -1.375e-01 & -1.081e-01 & 1.362e-12 \\ 
        ~ & 0.005 & 87.42 & -4.085e-02 & -1.375e-01 & -1.081e-01 & 1.379e-12 \\ \hline
        ~ & 0.1 & 2.96 & -4.085e-02 & -1.375e-01 & -1.081e-01 & 1.368e-12 \\ 
        ~ & 0.05 & 5.19 & -4.085e-02 & -1.375e-01 & -1.081e-01 & 1.269e-12 \\ 
        GFLM-BFSP & 0.01 & 25.68 & -4.085e-02 & -1.375e-01 & -1.081e-01 & 1.313e-12 \\ 
        ~ & 0.005 & 53.30 & -4.085e-02 & -1.375e-01 & -1.081e-01 & 1.297e-12 \\ 
        ~ & 0.001 & 268.60 & -4.085e-02 & -1.375e-01 & -1.081e-01 & 1.348e-12 \\ \hline
    \end{tabular}
\end{table}

From the results in Tables \ref{tab: ex4.1} and \ref{tab: ex4.2}, we observe the following:
\begin{enumerate}[(1)]
    \item \textbf{Efficiency:} The semi-implicit BFSP schemes are significantly faster than the fully iterative BESP counterparts.
    \item \textbf{Accuracy:} Standard GFDN-BFSP suffers from $O(\tau)$ splitting errors, leading to large residuals. In contrast, GFLM-BFSP corrects this via the Lagrange multiplier, achieving spectral accuracy comparable to BESP even with large time steps.
    \item \textbf{Stability:} The fixed-point iteration in BESP schemes diverges for large $\tau$ (e.g., $\tau=1$) due to the unbounded effective potential in quantum droplets, losing the unconditional stability found in repulsive BECs.
\end{enumerate}

Consequently, GFLM-BFSP is identified as the optimal solver, combining the efficiency of explicit schemes with the accuracy of implicit ones. We adopt it for all subsequent simulations.


\subsection{Validation of the density-locked model}
In Section~\ref{sec:eGPE}, we reduced the coupled two-component eGPE to a single-component density-locked model under the assumption of a fixed density ratio. In this subsection, we numerically verify the accuracy of this reduction by comparing its ground state solutions with those of the full two-component model. The time step size is taken as $\tau=0.001$ for subsequent experiments.

We define the relative errors in the wave function and the total energy as follows:
\begin{align*}
    \mathcal{E}_{\phi} = \frac{\|\phi_{DL} - \phi_{Full}\|_2}{\|\phi_{Full}\|_2} = \frac{\sqrt{\int |\phi_{DL} - \phi_{Full}|^2 \mathrm{d}\mathbf{x}}}{\sqrt{\int |\phi_{Full}|^2 \mathrm{d}\mathbf{x}}}, \quad
    \mathcal{E}_{E} = \frac{|E_{DL} - E_{Full}|}{|E_{Full}|},
\end{align*}
where $\phi_{DL}$ is the solution of the density-locked model, and $\phi_{Full}$ is the profile of the first component $\phi_1$ from the full model (normalized to 1). The parameters for the ${}^{39}\mathrm{K}$ mixture are the same as in Example 4.1.

\textbf{Example 4.3} (Accuracy test). 
We systematically investigate the validity of the density-locked model by comparing it with the full two-component model. We vary the trapping frequency $\omega$, particle number $N$, and the interaction imbalance parameter $\Delta a$. The ground states are computed using the GFLM-BFSP method. 

\textbf{(1) Dependence on confinement strength $\omega$.} 
We fix $N=5\times10^4$ and vary the isotropic trapping frequency $\omega$ from $0$ to $2000$. As shown in Fig. \ref{fig:accuracy_test}(a), the density-locked model exhibits excellent accuracy in the strong confinement regime. The energy error $\mathcal{E}_E$ drops rapidly as $\omega$ increases and stabilizes around $10^{-3}$. The wavefunction error $\mathcal{E}_{\phi}$ initially decreases, reaching a minimum around $\omega \approx 600$, and remains below $0.7\%$ throughout the tested range. This confirms that the single-component approximation is robust under tight confinement, where the density profiles are stiff.

\textbf{(2) Dependence on particle number $N$.} 
We fix $\omega=1000$ and vary $N$ from $1\times 10^4$ to $1\times 10^5$. The results in Fig. \ref{fig:accuracy_test}(b) show a monotonic increase in errors as the droplet size grows. This trend is expected because larger droplets have a flatter density core (quantum droplet regime), making the transition region more sensitive to the ratio-locking approximation compared to Gaussian-like profiles. However, even for macroscopic droplets with $N=10^5$, the relative errors remain negligible ($\mathcal{E}_E < 0.3\%$ and $\mathcal{E}_{\phi} < 0.6\%$), validating the model for large-scale simulations.

\textbf{(3) Dependence on interaction imbalance $\Delta a$.} 
The density-locked ansatz is derived based on a specific ratio determined by the intraspecies interactions. To test the robustness of this approximation against parameter variations, we introduce a deviation $\Delta a$ relative to our baseline experimental configuration (i.e., $a_{12} = a_{12}^{\text{base}} + \Delta a$). We fix $\omega=1000, N=5\times10^4$ and vary $\Delta a$ from $0$ to $3a_0$. As illustrated in Fig. \ref{fig:accuracy_test}(c), the wavefunction error $\mathcal{E}_{\phi}$ increases almost linearly with $\Delta a$. This is physically consistent: as the interspecies interaction shifts away from the baseline value used to justify the locking condition, the true ground-state density ratio slightly drifts, leading to a growing discrepancy with the locked ansatz. Nevertheless, even at a significant deviation of $\Delta a = 3a_0$, the errors remain well-controlled ($\mathcal{E}_{\phi} < 1\%$ and $\mathcal{E}_E < 4\%$), demonstrating that the model is sufficiently robust for parameter sweeps around the base configuration.

\begin{figure}[htbp]
    \centering
    \begin{subfigure}[b]{0.8\textwidth}
        \centering
        \includegraphics[width=\textwidth]{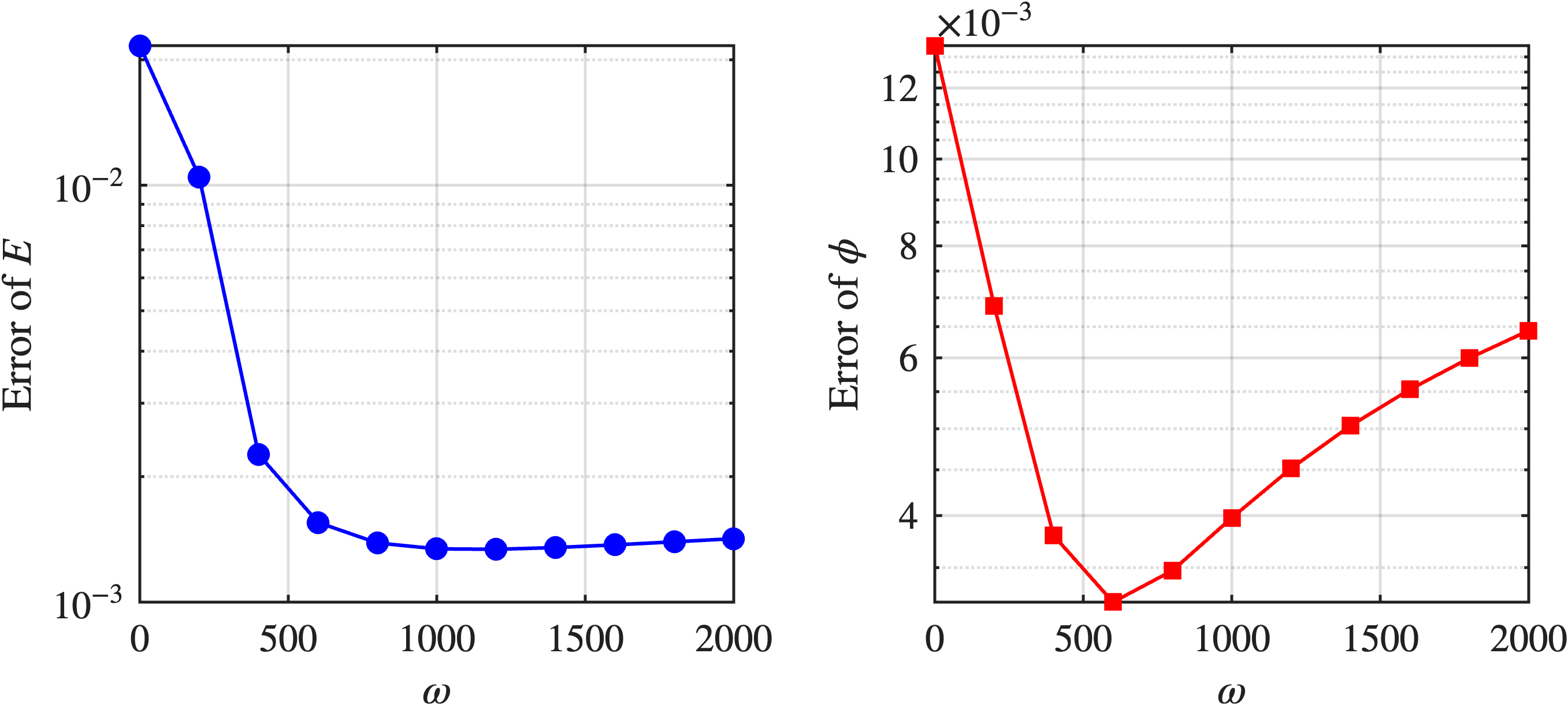}
        \caption{Error vs. confinement strength $\omega$ ($N=5\times10^4, \Delta a=0$).}
        \label{fig:err_omega}
    \end{subfigure}
    
    \vspace{1ex} 
    
    \begin{subfigure}[b]{0.8\textwidth}
        \centering
        \includegraphics[width=\textwidth]{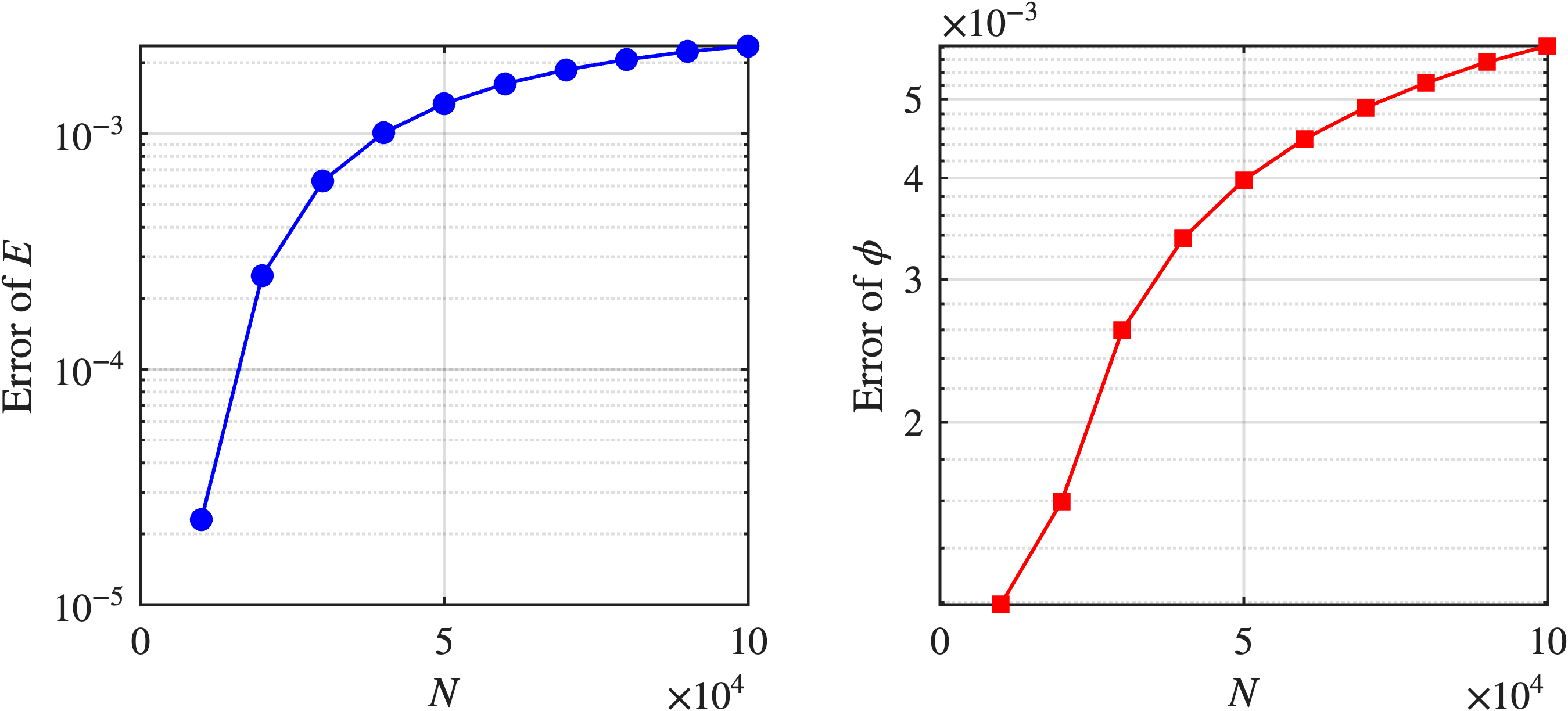}
        \caption{Error vs. particle number $N$ ($\omega=1000, \Delta a=0$).}
        \label{fig:err_N}
    \end{subfigure}
    
    \vspace{1ex}
    
    \begin{subfigure}[b]{0.8\textwidth}
        \centering
        \includegraphics[width=\textwidth]{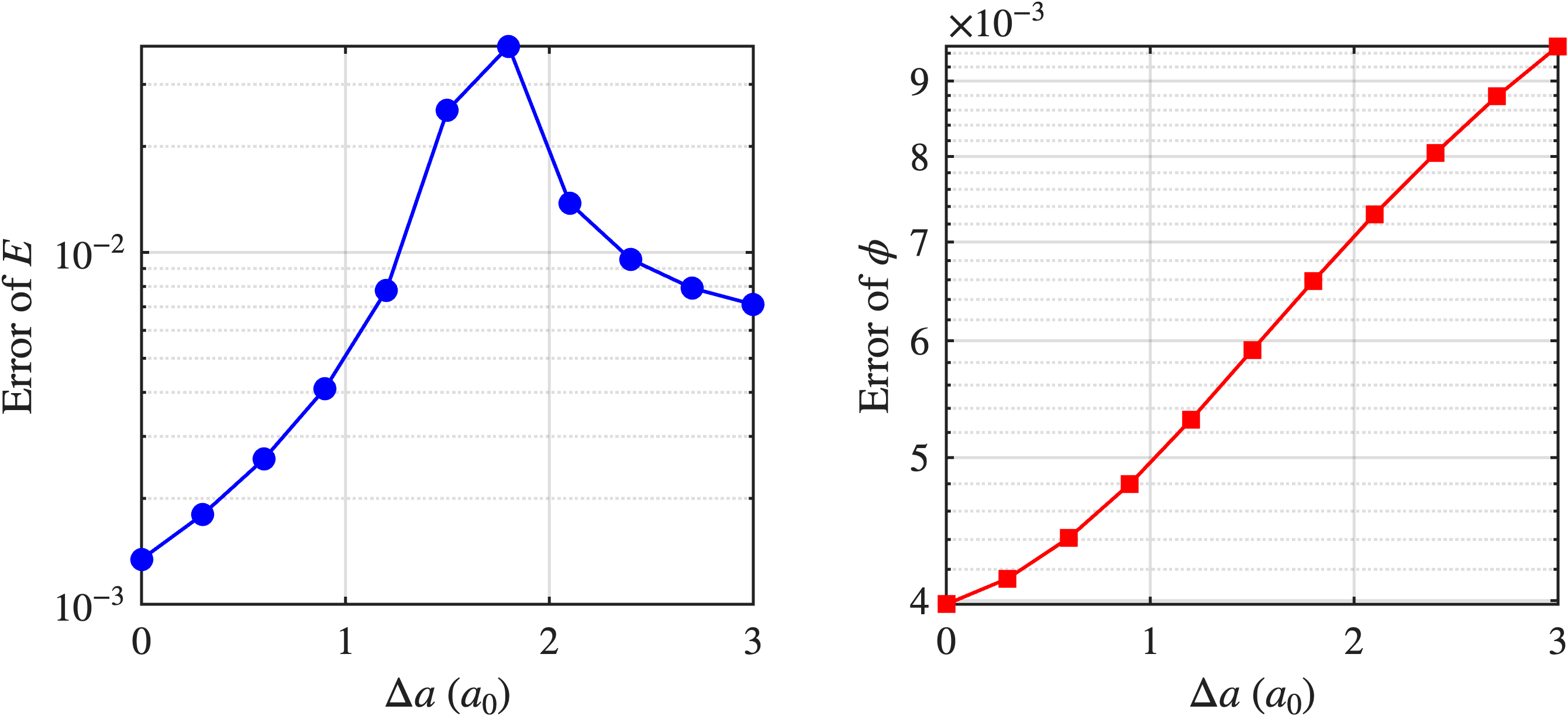}
        \caption{Error vs. interaction imbalance $\Delta a$ ($N=5\times10^4, \omega=1000$).}
        \label{fig:err_delta_a}
    \end{subfigure}
    
    \vspace{-1ex}
    \caption{Accuracy validation of the single-component density-locked model compared to the full two-component GPE. The left panels show the relative energy error $\mathcal{E}_E$, and the right panels show the relative $L^2$-norm wavefunction error $\mathcal{E}_{\phi}$. The results demonstrate that the reduced model maintains high precision ($\sim 10^{-3}$ relative error) across a wide range of experimentally relevant parameters.}
    \label{fig:accuracy_test}
\end{figure}

In summary, the density-locked model provides a reliable description of the ground state properties while significantly reducing the computational cost by halving the degrees of freedom.

\subsection{Ground state properties and validation of TFA}
In this subsection, we investigate the ground state properties of the density-locked model across different dimensions ($d=1, 2, 3$). We examine the transition from the weak-coupling to the strong-coupling regime by varying the particle number $N$. Our numerical study covers both free space ($V(\mathbf{x})=0$) and harmonic confinement, while the quantitative validation of the TFA focuses on the free-space limit.

We compute the ground state wave function $\phi_g$, the chemical potential $\mu_g$, and the root mean square (RMS) radius $r_{\mathrm{rms}}$ defined by $r_{\mathrm{rms}} = \sqrt{\int_{\mathcal{D}} |\mathbf{x}|^2 |\phi_g(\mathbf{x})|^2 \mathrm{d}\mathbf{x}}$. 
The computational domain $\mathcal{D}$ is chosen large enough to avoid boundary effects: $[-64, 64]$ for 1D, $[-32, 32]^2$ for 2D, and $[-16, 16]^3$ for 3D, with sufficiently fine meshes to ensure spatial accuracy.

\subsubsection*{Transition to the flat-top profile}
Figs. \ref{fig: 1d}, \ref{fig: 2d}, and \ref{fig: 3d} display the ground state profiles for varying particle numbers $N$ in both free space and harmonic traps. A consistent trend is observed across all dimensions:
\begin{enumerate}[(1)]
    \item For small $N$, the kinetic energy is comparable to the interaction energy, resulting in smooth, Gaussian-like profiles.
    \item As $N$ increases, the droplet expands (increasing $r_{\mathrm{rms}}$) and the peak density decreases.
    \item For sufficiently large $N$, the kinetic energy becomes negligible. In free space, the wave function develops a distinct ``flat-top" structure with a uniform bulk density matching the TFA prediction $\rho_s = -5\alpha/6\beta$. Under harmonic confinement, the profile broadens into a TF distribution determined by the local potential.
\end{enumerate}

Tables \ref{table: 1d}, \ref{table: 2d}, and \ref{table: 3d} list the quantitative properties for the free-space cases. We observe that the TFA provides an accurate estimate for $\mu_g$ only when $N$ is very large. For intermediate $N$, finite-size effects and surface tension 
lead to deviations.

\subsubsection*{Convergence rate of TFA}
To quantify the accuracy of the TFA in the limit $N \to \infty$ (free space), we calculate the errors between the numerical ground state $(\phi_g, \mu_g)$ and the TF approximation $(\phi^s, \mu^s)$. The convergence results are summarized in Tables \ref{tab: conv_1d}, \ref{tab: conv_2d}, and \ref{tab: conv_3d}.

We define the error with respect to the small parameter $N$. Based on the numerical data, we observe the following convergence laws:
\begin{itemize}
    \item Chemical Potential: The error $|\mu_g - \mu^s|$ converges at a rate of approximately $O(N^{-1/d})$.
    \item Wave Function ($L^2$ norm): The error $\|\phi_g - \phi^s\|_{L^2}$ converges at a rate of approximately $O(N^{-1/(2d)})$.
\end{itemize}
Specifically:
\begin{itemize}
    \item In 1D (Table \ref{tab: conv_1d}), the $L^2$ error decays as $O(N^{-0.5})$.
    \item In 2D (Table \ref{tab: conv_2d}), the chemical potential error scales as $O(N^{-0.5})$, and the $L^2$ error scales as $O(N^{-0.25})$.
    \item In 3D (Table \ref{tab: conv_3d}), the chemical potential error scales as $O(N^{-0.35}) \!\approx\! O(N^{-1/3})$, and the $L^2$ error scales as $O(N^{-0.17}) \approx O(N^{-1/6})$.
\end{itemize}
These scaling laws suggest that the error is dominated by the surface layer of the droplet. The convergence rates reflect how the ratio of the surface layer volume to the total droplet volume scales with $N$ in different dimensions. While these numerical findings align with physical intuition, establishing a rigorous mathematical proof for these dimension-dependent convergence rates remains an open problem for future research.


\begin{figure}[!t]
    \centering
    \begin{subfigure}[t]{0.48\linewidth}
        \centering
        \includegraphics[width=\textwidth]{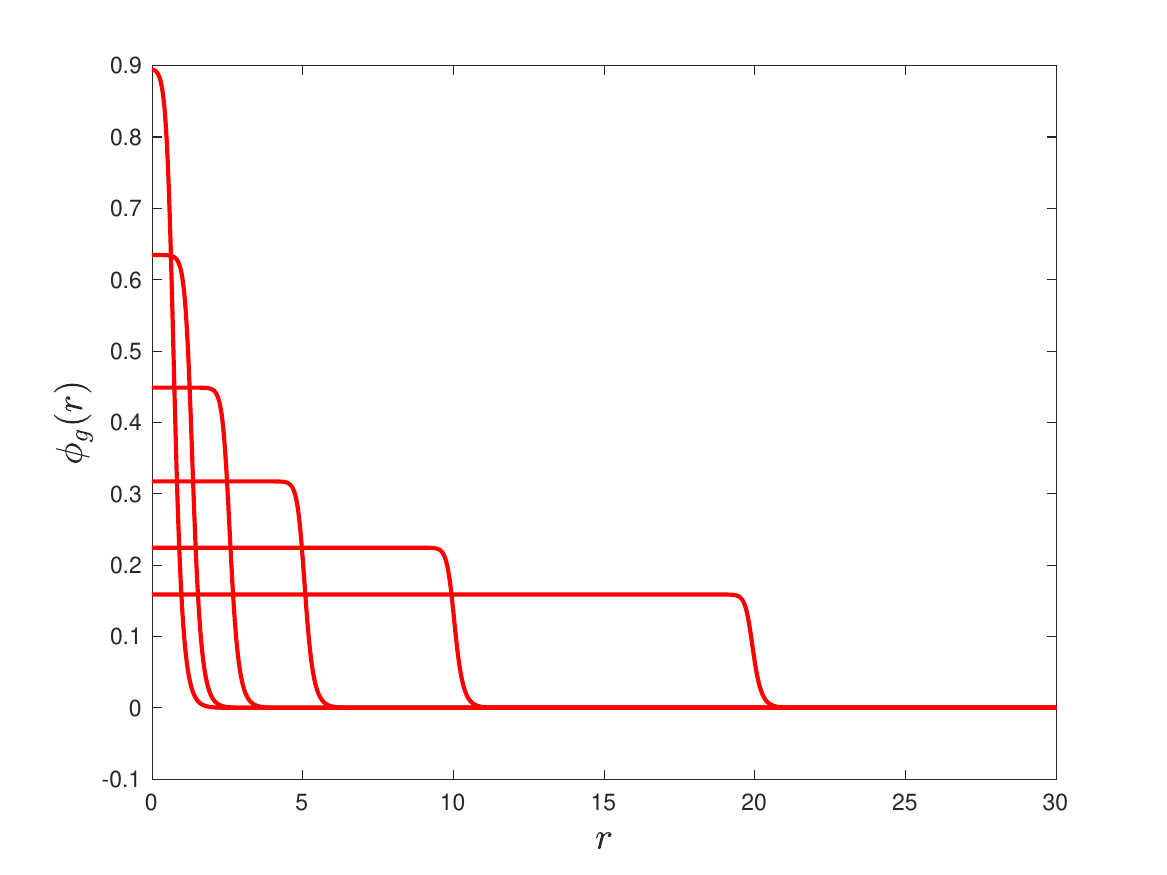}
        \caption{Free space ($V(x)=0$)}
    \end{subfigure}
    \hfill
    \begin{subfigure}[t]{0.48\linewidth}
        \centering
        \includegraphics[width=\textwidth]{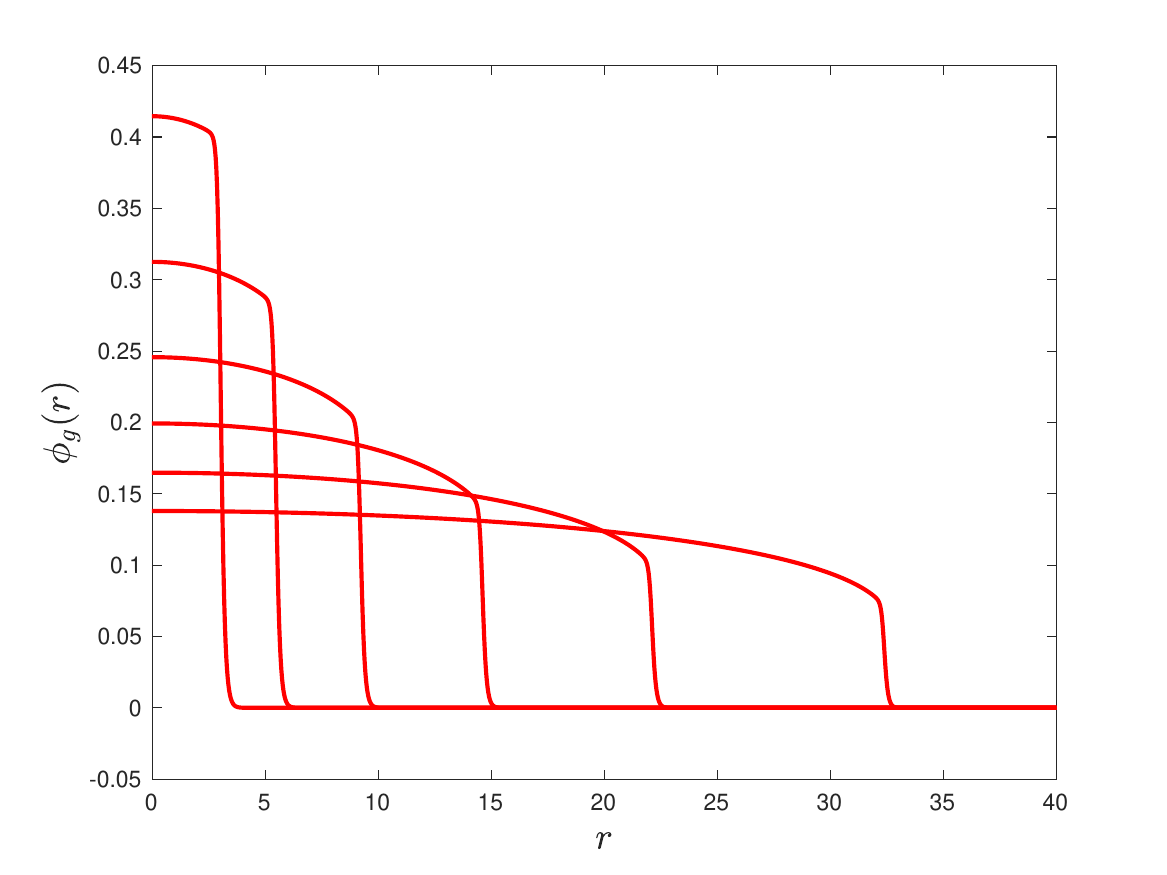}
        \caption{Harmonic potential ($V(x)=\frac{1}{2}\omega_r x^2$)}
    \end{subfigure}
    \caption{Ground state profiles of 1D density-locked droplets. (a) Wave function $\phi_g(r)$ in free space with scaling parameters $\omega=50$ and $\omega_{\perp}=100$. (b) Wave function $\phi_g(r)$ under a harmonic trap with $\omega_r=20$ and $\omega_{\perp}=2000$. In both panels, the six curves correspond to particle numbers $N$ starting from $2.5 \times 10^4$ and doubling successively up to $8 \times 10^5$. The profiles with wider spatial support correspond to larger particle numbers.}
    \label{fig: 1d}
    \vspace{-1ex}
\end{figure}

\begin{table}[!t]
    \centering
    \caption{Ground state properties of 1D quantum droplets in free space. $\phi_g(0)$ denotes the peak amplitude, $r_{\mathrm{rms}}$ is the RMS radius, and $\mu_g$ is the chemical potential computed numerically compared with the TFA.} 
    \label{table: 1d}
    \small 
    \begin{tabular}{lcccc}
    \hline
        $N$ & $\phi_{g}(0)$ & $r_{\mathrm{rms}}$ & $\mu_g$ (Num.) & $\mu_g$ (TFA) \\ \hline
        $6\times10^3$    & 1.3735 & 0.2208 & -11.4921 & -13.6051 \\ 
        $1.2\times10^4$  & 1.2138 & 0.2553 & -13.4526 & -13.6051 \\ 
        $2.4\times10^4$  & 0.9121 & 0.3923 & -13.6044 & -13.6051 \\ 
        $4.8\times10^4$  & 0.6476 & 0.7137 & -13.6051 & -13.6051 \\ 
        $9.6\times10^4$  & 0.4580 & 1.3894 & -13.6051 & -13.6051 \\ 
        $1.92\times10^5$ & 0.3238 & 2.7594 & -13.6051 & -13.6051 \\ 
        $3.84\times10^5$ & 0.2290 & 5.5089 & -13.6051 & -13.6051 \\ 
        $7.68\times10^5$ & 0.1619 & 11.0130 & -13.6051 & -13.6051 \\ \hline
    \end{tabular}
\end{table}

\begin{figure}[!t]
    \centering
    \begin{subfigure}[t]{0.48\linewidth}
        \centering
        \includegraphics[width=\textwidth]{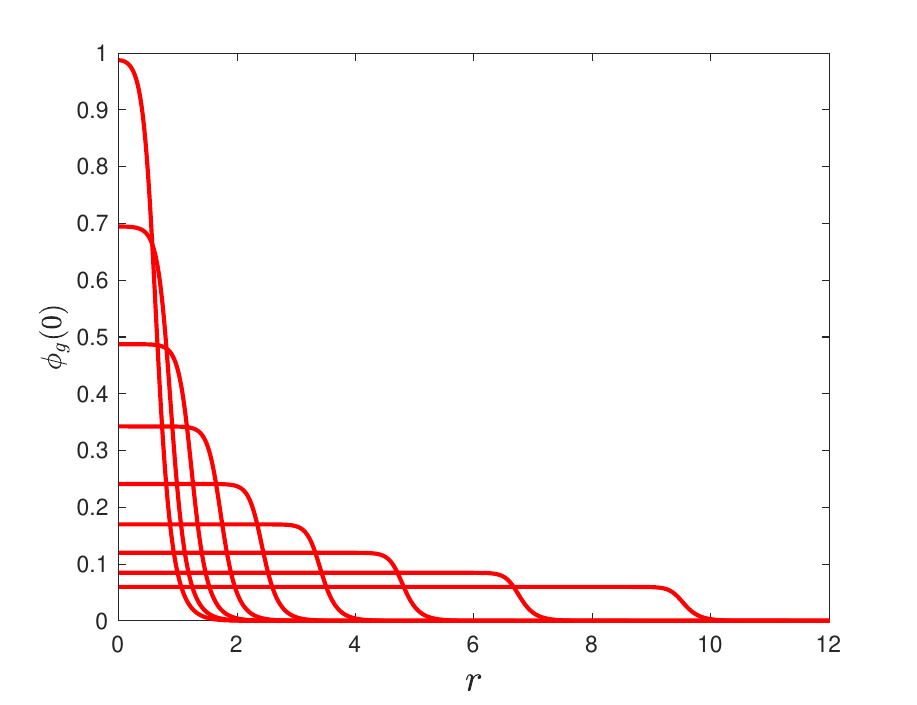}
        \caption{Free space ($V(r)=0$)}
    \end{subfigure}
    \hfill
    \begin{subfigure}[t]{0.48\linewidth}
        \centering
        \includegraphics[width=\textwidth]{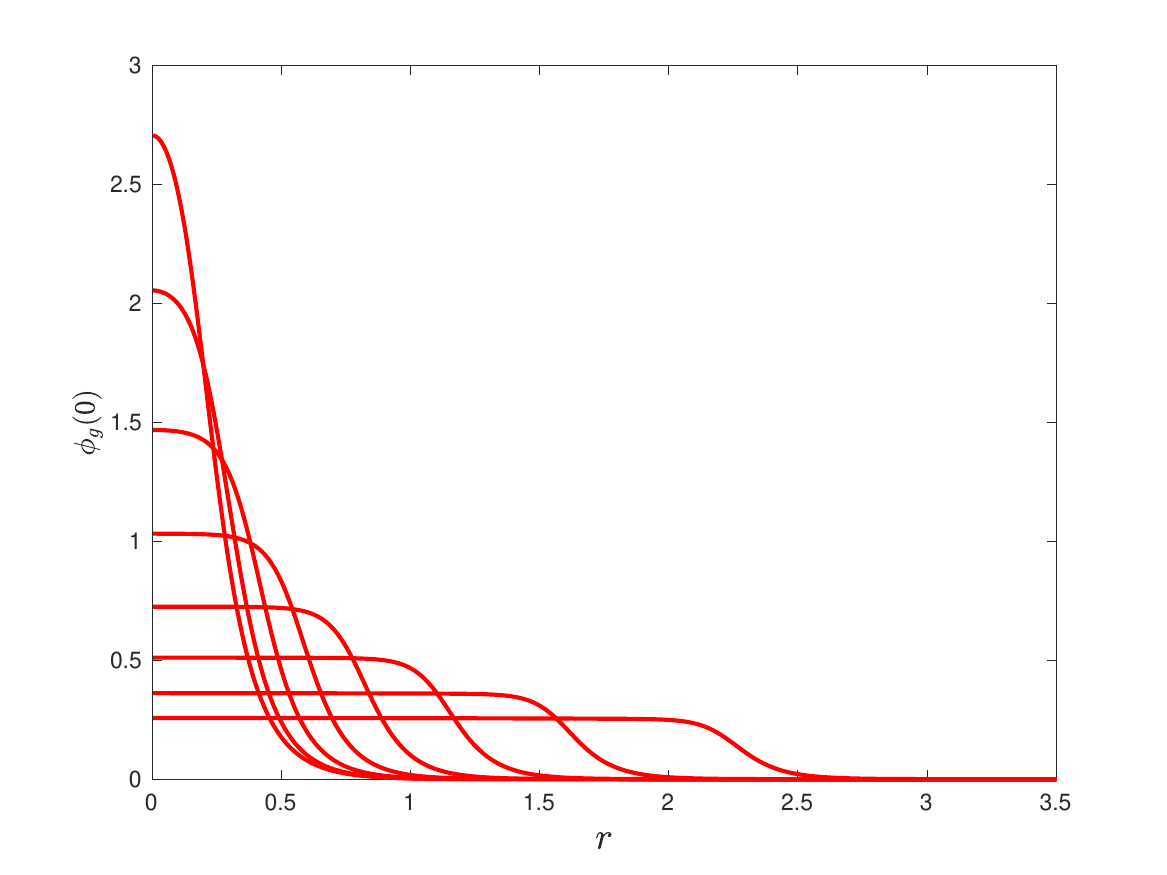}
        \caption{Harmonic potential ($V(r)=\frac{1}{2}\omega_r r^2$)}
    \end{subfigure}
    \caption{Ground state profiles of 2D radially symmetric droplets. (a) Radial wave function $\phi_g(r)$ in free space with scaling parameters $\omega=50$ and $\omega_z=100$. (b) Radial wave function $\phi_g(r)$ under a harmonic trap with $\omega_r=20$ and $\omega_z=2000$. The curves correspond to particle numbers $N$ starting from $10^5$ and doubling successively up to $2.56\times10^7$ in (a), and from $10^4$ up to $1.28\times10^6$ in (b). The profiles with wider spatial support correspond to larger particle numbers.}
    \label{fig: 2d}
\end{figure}

\begin{table}[!t]
    \centering
    \caption{Ground state properties of 2D quantum droplets in free space. The notations are the same as in Table \ref{table: 1d}.}
    \label{table: 2d}
    \small
    \begin{tabular}{lcccc}
    \hline
        $N$ & $\phi_g(0)$ & $r_{\mathrm{rms}}$ & $\mu_g$ (Num.) & $\mu_g$ (TFA) \\ 
    \hline
        $1\times10^5$    & 0.9874 & 0.4640 & -13.1268 & -15.3924 \\ 
        $2\times10^5$    & 0.6946 & 0.6190 & -13.8162 & -15.3924 \\ 
        $4\times10^5$    & 0.4873 & 0.8503 & -14.2864 & -15.3924 \\ 
        $8\times10^5$    & 0.3425 & 1.1868 & -14.6134 & -15.3924 \\ 
        $16\times10^5$   & 0.2411 & 1.6696 & -14.8426 & -15.3924 \\ 
        $32\times10^5$   & 0.1699 & 2.3575 & -15.0040 & -15.3924 \\ 
        $64\times10^5$   & 0.1199 & 3.3341 & -15.1179 & -15.3924 \\ 
        $128\times10^5$  & 0.0846 & 4.7181 & -15.1984 & -15.3924 \\ 
        $256\times10^5$  & 0.0598 & 6.6772 & -15.2552 & -15.3924 \\ \hline
    \end{tabular}
\end{table}

\begin{figure}[!t]
    \centering
    \begin{subfigure}[t]{0.48\linewidth}
        \centering
        \includegraphics[width=\textwidth]{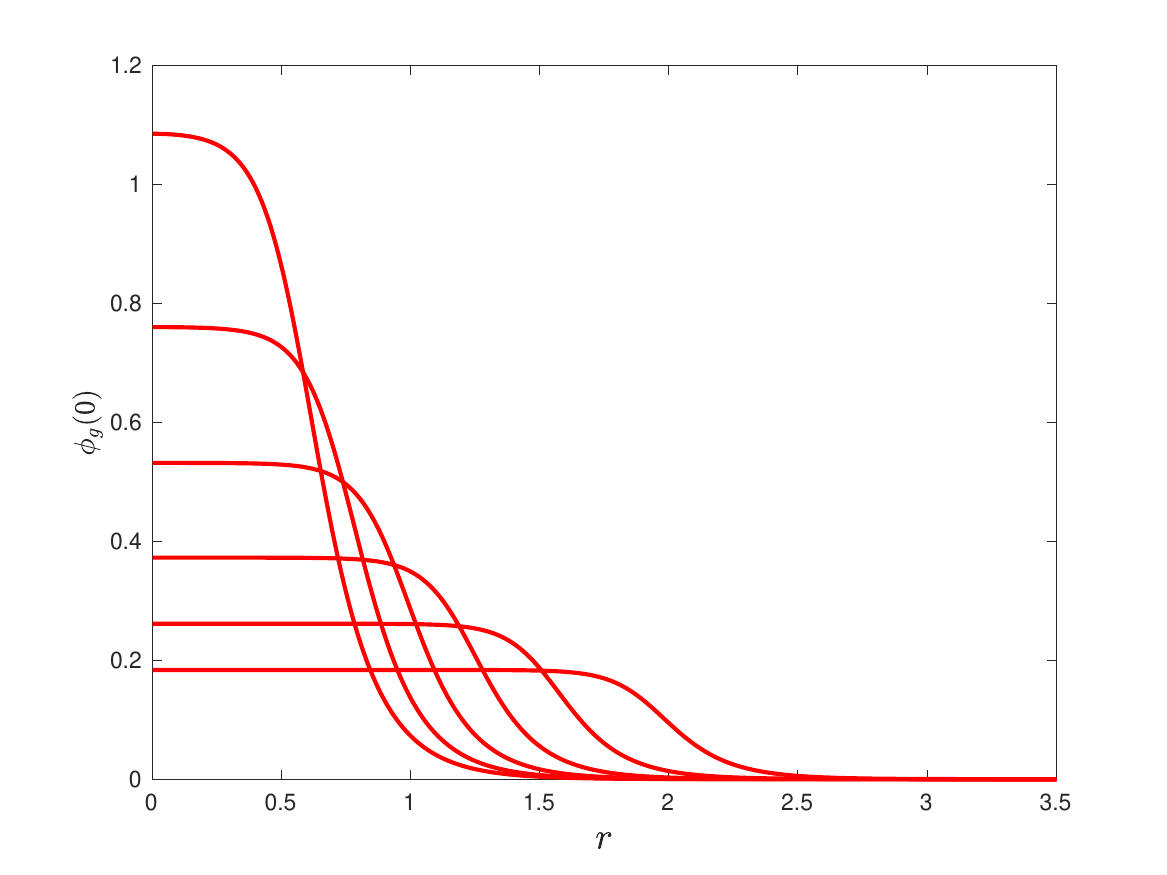}
        \caption{Free space ($V(r)=0$)}
    \end{subfigure}
    \hfill
    \begin{subfigure}[t]{0.48\linewidth}
        \centering
        \includegraphics[width=\textwidth]{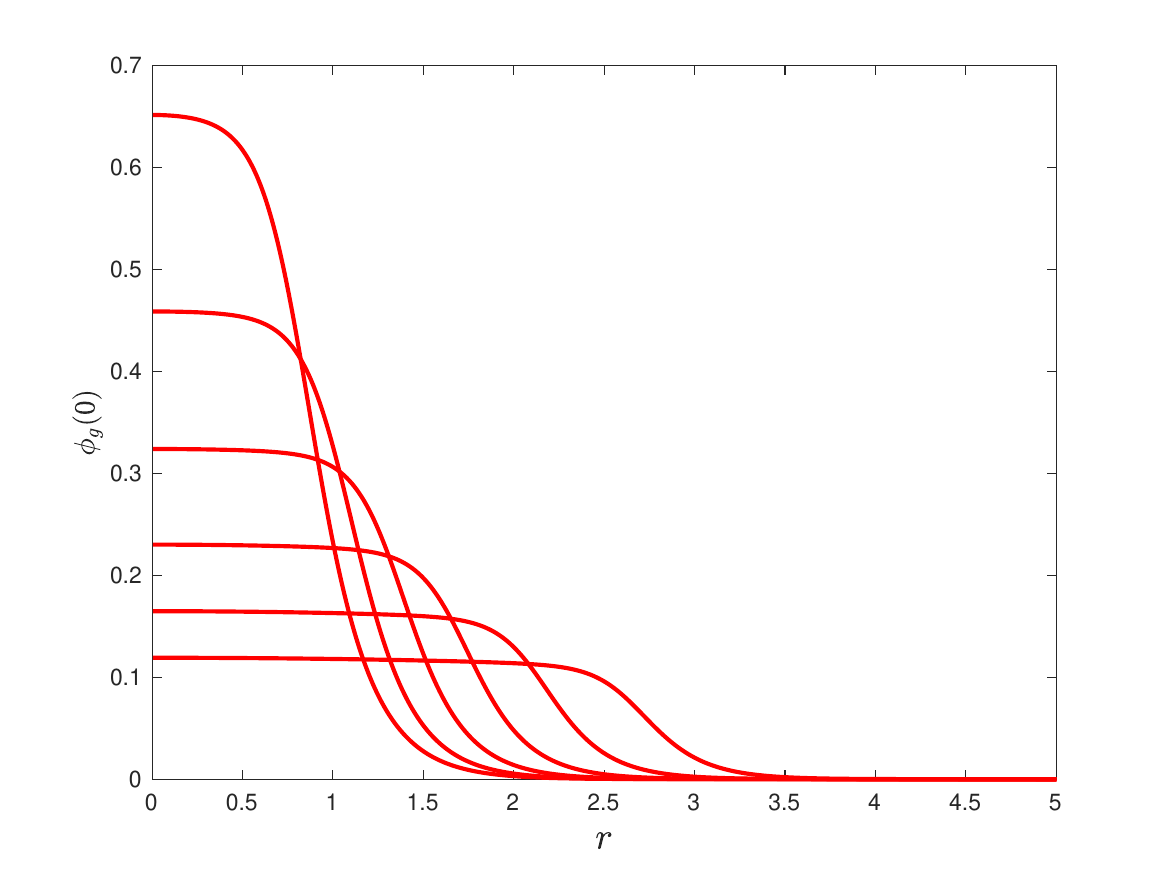}
        \caption{Harmonic potential ($V(r)=\frac{1}{2}\omega_r r^2$)}
    \end{subfigure}
    \caption{Ground state profiles of 3D spherically symmetric droplets. (a) Radial wave function $\phi_g(r)$ in free space with scaling parameter $\omega=50$. (b) Radial wave function $\phi_g(r)$ under a harmonic trap with $\omega_r=100$. In both panels, the curves correspond to particle numbers $N$ starting from $4\times10^5$ and doubling successively up to $1.28\times 10^7$. The profiles with wider spatial support correspond to larger particle numbers.}
    \label{fig: 3d}
\end{figure}

\begin{table}[!t]
    \centering
    \caption{Ground state properties of 3D quantum droplets in free space. The notations are the same as in Table \ref{table: 1d}.}
    \label{table: 3d}
    \small
    \begin{tabular}{lcccc}
    \hline
        $N$ & $\phi_g(0)$ & $r_{\mathrm{rms}}$ & $\mu_g$ (Num.) & $\mu_g$ (TFA) \\ \hline
        $1\times10^5$    & 2.1070 & 0.4047 & -8.6447  & -17.4145 \\ 
        $2\times10^5$    & 1.5357 & 0.4511 & -10.7828 & -17.4145 \\ 
        $4\times10^5$    & 1.0854 & 0.5285 & -12.3262 & -17.4145 \\ 
        $8\times10^5$    & 0.7603 & 0.6369 & -13.4783 & -17.4145 \\ 
        $1.6\times10^6$  & 0.5319 & 0.7807 & -14.3525 & -17.4145 \\ 
        $3.2\times10^6$  & 0.3726 & 0.9675 & -15.0227 & -17.4145 \\ 
        $6.4\times10^6$  & 0.2614 & 1.2073 & -15.5402 & -17.4145 \\ 
        $1.28\times10^7$ & 0.1837 & 1.5128 & -15.9419 & -17.4145 \\ 
        $10^8$           & 0.0649 & 2.9888 & -16.6868 & -17.4145 \\ \hline
    \end{tabular}
\end{table}


\begin{table}[!t]
    \centering
    \caption{Convergence of the TFA in 1D as $N \to \infty$.}
    \label{tab: conv_1d}
    \resizebox{\textwidth}{!}{
    \begin{tabular}{llllllll}
    \hline
       $N$ & 6E3 & 1.2E4 & 2.4E4 & 4.8E4 & 9.6E4 & Rate \\ \hline
       $\|\phi_g-\phi^s\|_{L^2}$ & 7.492e-01 & 5.082e-01 & 3.580e-01 & 2.547e-01 & 1.823e-01 & $\approx 0.49$ \\ \hline
    \end{tabular}}
\end{table}

\begin{table}[!t]
    \centering
    \caption{Convergence of the TFA in 2D as $N \to \infty$.}
    \label{tab: conv_2d}
    \resizebox{\textwidth}{!}{
    \begin{tabular}{llllllll}
    \hline
       $N$ & 1E5 & 2E5 & 4E5 & 8E5 & 16E5 & 32E5 & Rate \\ \hline
       $|\mu_g-\mu^s|$  & 2.266 & 1.576 & 1.106 & 7.792e-1 & 5.499e-1 & 3.885e-1 & $\approx 0.50$ \\ 
       $\|\phi_g-\phi^s\|_{L^2}$ & 4.623e-1 & 3.900e-1 & 3.268e-1 & 2.735e-1 & 2.312e-1 & 1.945e-1 & $\approx 0.25$ \\ \hline
    \end{tabular}}
\end{table}

\begin{table}[!t]
    \centering
    \caption{Convergence of the TFA in 3D as $N \to \infty$.}
    \label{tab: conv_3d}
    \resizebox{\textwidth}{!}{
    \begin{tabular}{llllllll}
    \hline
       $N$ & 1E5 & 2E5 & 4E5 & 8E5 & 16E5 & 32E5 & Rate \\ \hline
       $|\mu_g-\mu^s|$  & 8.771 & 6.633 & 5.089 & 3.937 & 3.063 & 2.392 & $\approx 0.35$ \\ 
       $\|\phi_g-\phi^s\|_{L^2}$ & 6.929e-1 & 6.000e-1 & 5.297e-1 & 4.691e-1 & 4.154e-1 & 3.697e-1 & $\approx 0.17$ \\ \hline
    \end{tabular}}
\end{table}

\subsection{The critical particle number}
As established in Theorem \ref{Thm: Existence of ground state}, a critical particle number $N_c$ exists for the formation of a self-bound droplet in free space. When $N < N_c$, the attractive mean-field energy is insufficient to bind the atoms against the LHY repulsion and kinetic pressure, resulting in the non-existence of a ground state. While theoretical estimates exist based on variational approximations (e.g., using Gaussian ansatz) \cite{petrov_quantum_2015}, the exact value of $N_c$ remains to be determined numerically. In this subsection, we provide a precise determination of $N_c$ using the GFLM-BFSP method.

To find the universal critical threshold independent of specific scattering lengths, we introduce a rescaled particle number $\widetilde{N}$. By choosing the scaling frequency as 
in \eqref{eq:omega-free-space}, 
the dimensionless coefficients become $\alpha = -3\widetilde{N}$ and $\beta = \frac{5}{2}\widetilde{N}^{3/2}$, where
\begin{equation} \label{eq: N_tilde_def}
    \widetilde{N} = \frac{5\pi^2}{3\sqrt{6}} \frac{|\delta a|^{5/2}}{(\sqrt{a_{11}}+\sqrt{a_{22}})^5} N.
\end{equation}
Thus, the problem is reduced to finding a critical parameter $\widetilde{N}_c$ such that a negative energy ground state exists if and only if $\widetilde{N} > \widetilde{N}_c$.

We perform numerical simulations by varying $\widetilde{N}$ (via $N$) and monitoring the final energy $E_{\text{fin}}$ at the termination of the algorithm. To avoid convergence to spurious positive-energy excited states (which would incorrectly imply the non-existence of a bound state), we employ a numerical continuation strategy: starting from a stable large-$\widetilde{N}$ solution, we progressively decrease the particle number using the previous wave function as the initial guess. This ensures the solver remains within the basin of attraction of the true ground state (as long as it exists). The transition is thus accurately identified where the computed $E_{\text{fin}}$ crosses from negative to positive. Table \ref{table: critical atomic number} presents the results near the transition point.

\begin{table}[!t]
    \centering
    \caption{Final energy $E_{\text{fin}}$ near the critical parameter $\widetilde{N}_c$. The transition from stable droplet ($E_{\text{fin}} < 0$) to instability occurs around $\widetilde{N} \approx 22.65$.}
    \label{table: critical atomic number}
    \small
    \begin{tabular}{ccccccc}
    \hline
       $\widetilde{N}$ & $22.60$ & $22.62$ & $22.64$ & $22.66$ & $22.68$ & $22.70$ \\ \hline
       $E_{\text{fin}}$ & 2.667e-4 & 1.498e-4 & 3.280e-5 & -8.420e-5  & -2.010e-4 & -3.179e-4 \\ \hline
    \end{tabular}
\end{table}
Our numerical results indicate that the critical threshold is $\widetilde{N}_c \approx 22.65$. This value is significantly higher than the analytical estimate $\widetilde{N}_c \approx 18.65$ provided by Petrov \cite{petrov_quantum_2015}, which was derived under a Gaussian variational approximation. Our result suggests that the Gaussian ansatz underestimates the critical number because it fails to capture the flat-top profile of the droplet near the transition.

Substituting $\widetilde{N}_c \approx 22.65$ into Eq. \eqref{eq: N_tilde_def} yields the physical critical particle number:
\begin{equation}\label{critical number with scattering length}
    N_c \approx 3.373 \times \frac{(\sqrt{a_{11}}+\sqrt{a_{22}})^5}{|\delta a|^{5/2}}. 
\end{equation}
This relationship provides a precise criterion for the formation of quantum droplets in homonuclear Bose-Bose mixtures.

\section{Summary}\label{sec:Summary}
In this paper, we have systematically studied the ground state computation of quantum droplets in homonuclear Bose-Bose mixtures. We derived the dimensionless energy functionals for both the general two-component system and the reduced density-locked model, providing a clear theoretical framework for numerical simulations.
To compute the ground states efficiently, we adapted and compared several gradient flow discretization schemes. Our extensive numerical benchmarks demonstrate that the GFLM-BFSP method achieves the best balance between efficiency and accuracy. 
Based on this efficient solver, we conducted a series of numerical experiments to explore the physical properties of quantum droplets. We verified the validity of the density-locked model, showing it to be an accurate approximation for calculating ground state properties. We also provided a quantitative analysis of the Thomas-Fermi approximation, establishing its dimension-dependent convergence rates in the strong-coupling regime. Furthermore, we numerically determined the critical particle number for self-binding in free space, offering a precise correction to existing analytical estimates.
This work provides a robust basis for future studies. Natural extensions include the simulation of droplet dynamics, the investigation of rotational properties and vortex formation, and the study of droplets in heteronuclear mixtures.

\section*{Acknowledgments}
We would like to thank Prof. Weizhu Bao for helpful discussions. 
We thank Westlake University HPC Center for computation support. 

\bibliographystyle{siamplain}
\bibliography{references}

@article{ALT2017JCP,
  author  = {Xavier Antoine and Antoine Levitt and Qinglin Tang},
  title   = {Efficient spectral computation of the stationary states of rotating {B}ose-{E}instein condensates by preconditioned nonlinear conjugate gradient methods},
  journal = {J. Comput. Phys.},
  volume  = {343},
  pages   = {92-109},
  year    = {2017},
  issn    = {0021-9991}
}

@article{BC2013KRM,
  author  = {W. Bao and Y. Cai},
  title   = {Mathematical theory and numerical methods for {B}ose-{E}instein condensation},
  journal = {Kinet. Relat. Mod.},
  year    = {2013},
  volume  = {6},
  number  = {1},
  pages   = {1-135}
}

@article{Cances10,
  title   = {Numerical Analysis of Nonlinear Eigenvalue Problems},
  author  = {Canc\`{e}s, E. and Chakir, R. and Maday, Y.},
  journal = {J. Sci. Comput.},
  volume  = {45},
  pages   = {90-117},
  year    = {2010}
}

@article{CST2000PRE,
  title   = {Ground state of trapped interacting {B}ose-{E}instein condensates by an explicit imaginary-time algorithm},
  author  = {M. L. Chiofalo and S. Succi and M. P. Tosi},
  journal = {Phys. Rev. E},
  year    = {2000},
  volume  = {62},
  number  = {5},
  pages   = {7438-7444}
}

@article{CDLX2023JCP,
title = {Second-order flows for computing the ground states of rotating {Bose-Einstein} condensates},
journal = {J. Comput. Phys.},
volume = {475},
pages = {111872},
year = {2023},
author = {Haifan Chen and Guozhi Dong and Wei Liu and Ziqing Xie}
}

@article{DP2017SISC,
  author  = {Ionut Danaila and Bartosz Protas},
  title   = {Computation of Ground States of the {G}ross-{P}itaevskii Functional via {R}iemannian Optimization},
  journal = {SIAM J. Sci. Comput.},
  volume  = {39},
  number  = {6},
  pages   = {B1102-B1129},
  year    = {2017}
}

@article{HJ2025SIREV,
author = {Henning, Patrick and Jarlebring, Elias},
title = {The {Gross--Pitaevskii} Equation and Eigenvector Nonlinearities: Numerical Methods and Algorithms},
journal = {SIAM Rev.},
volume = {67},
number = {2},
pages = {256-317},
year = {2025}
}

@article{ancilotto_self-bound_2018,
  title = {Self-bound ultradilute {Bose} mixtures within local density approximation},
  volume = {98},
  number = {5},
  journal = {Phys. Rev. A},
  author = {Ancilotto, Francesco and Barranco, Manuel and Guilleumas, Montserrat and Pi, Martí},
  month = nov,
  year = {2018},
  pages = {053623}
}

@article{minardi_effective_2019,
  title = {Effective expression of the {Lee}-{Huang}-{Yang} energy functional for heteronuclear mixtures},
  volume = {100},
  number = {6},
  journal = {Phys. Rev. A},
  author = {Minardi, F. and Ancilotto, F. and Burchianti, A. and D'Errico, C. and Fort, C. and Modugno, M.},
  month = dec,
  year = {2019},
  pages = {063636}
}

@article{semeghini_self-bound_2018,
  title = {Self-{Bound} {Quantum} {Droplets} of {Atomic} {Mixtures} in {Free} {Space}},
  volume = {120},
  number = {23},
  journal = {Phys. Rev. Lett.},
  author = {Semeghini, G. and Ferioli, G. and Masi, L. and Mazzinghi, C. and Wolswijk, L. and Minardi, F. and Modugno, M. and Modugno, G. and Inguscio, M. and Fattori, M.},
  month = jun,
  year = {2018},
  pages = {235301}
}

@article{derrico_observation_2019,
  title = {Observation of quantum droplets in a heteronuclear bosonic mixture},
  volume = {1},
  number = {3},
  journal = {Phys. Rev. Research},
  author = {D'Errico, C. and Burchianti, A. and Prevedelli, M. and Salasnich, L. and Ancilotto, F. and Modugno, M. and Minardi, F. and Fort, C.},
  month = dec,
  year = {2019},
  pages = {033155}
}

@article{ferioli_dynamical_2020,
  title = {Dynamical formation of quantum droplets in a ${}^{39}\mathrm{K}$ mixture},
  volume = {2},
  number = {1},
  journal = {Phys. Rev. Research},
  author = {Ferioli, G. and Semeghini, G. and Terradas-Briansó, S. and Masi, L. and Fattori, M. and Modugno, M.},
  month = mar,
  year = {2020},
  pages = {013269}
}

@article{cheiney_bright_2018,
  title = {Bright {Soliton} to {Quantum} {Droplet} {Transition} in a {Mixture} of {Bose}-{Einstein} {Condensates}},
  volume = {120},
  number = {13},
  journal = {Phys. Rev. Lett.},
  author = {Cheiney, P. and Cabrera, C. R. and Sanz, J. and Naylor, B. and Tanzi, L. and Tarruell, L.},
  month = mar,
  year = {2018},
  pages = {135301}
}

@article{flynn_quantum_2023,
  title = {Quantum droplets in imbalanced atomic mixtures},
  volume = {5},
  number = {3},
  journal = {Phys. Rev. Research},
  author = {Flynn, T. A. and Parisi, L. and Billam, T. P. and Parker, N. G.},
  month = sep,
  year = {2023},
  pages = {033167}
}

@article{cappellaro_collective_2018,
  title = {Collective modes across the soliton-droplet crossover in binary {Bose} mixtures},
  volume = {97},
  number = {5},
  journal = {Phys. Rev. A},
  author = {Cappellaro, Alberto and Macrì, Tommaso and Salasnich, Luca},
  month = may,
  year = {2018},
  pages = {053623}
}

@article{shamriz_suppression_2020,
  title = {Suppression of the quasi-two-dimensional quantum collapse in the attraction field by the {Lee}-{Huang}-{Yang} effect},
  volume = {101},
  number = {6},
  journal = {Phys. Rev. A},
  author = {Shamriz, Elad and Chen, Zhaopin and Malomed, Boris A.},
  month = jun,
  year = {2020},
  pages = {063628}
}

@article{petrov_ultradilute_2016,
  title = {Ultradilute Low-Dimensional Liquids},
  volume = {117},
  number = {10},
  journal = {Phys. Rev. Lett.},
  author = {Petrov, D. S. and Astrakharchik, G. E.},
  year = {2016},
  pages = {100401}
}

@article{petrov_quantum_2015,
  title = {Quantum Mechanical Stabilization of a Collapsing {Bose}-{Bose} Mixture},
  volume = {115},
  number = {15},
  journal = {Phys. Rev. Lett.},
  author = {Petrov, D. S.},
  month = oct,
  year = {2015},
  pages = {155302}
}

@article{ferrier-barbut_observation_2016,
  title = {Observation of {Quantum} {Droplets} in a {Strongly} {Dipolar} {Bose} {Gas}},
  volume = {116},
  number = {21},
  journal = {Phys. Rev. Lett.},
  author = {Ferrier-Barbut, Igor and Kadau, Holger and Schmitt, Matthias and Wenzel, Matthias and Pfau, Tilman},
  month = may,
  year = {2016},
  pages = {215301}
}

@article{bao2006efficient,
  title={Efficient and spectrally accurate numerical methods for computing ground and first excited states in {Bose}--{Einstein} condensates},
  author={Bao, Weizhu and Chern, I-Liang and Lim, Fong Yin},
  journal={J. Comput. Phys.},
  volume={219},
  number={2},
  pages={836--854},
  year={2006}
}

@article{liu2021normalized,
  title={Normalized Gradient Flow with Lagrange Multiplier for Computing Ground States of {Bose}--{Einstein} Condensates},
  author={Liu, Wei and Cai, Yongyong},
  journal={SIAM J. Sci. Comput.},
  volume={43},
  number={1},
  pages={B219--B242},
  year={2021}
}

@article{cabrera2018quantum,
  title={Quantum liquid droplets in a mixture of {Bose}-{Einstein} condensates},
  author={Cabrera, CR and Tanzi, L and Sanz, J and Naylor, B and Thomas, P and Cheiney, P and Tarruell, L},
  journal={Science},
  volume={359},
  number={6373},
  pages={301--304},
  year={2018}
}

@article{faou2018convergence,
  title={Convergence of a normalized gradient algorithm for computing ground states},
  author={Faou, Erwan and J{\'e}z{\'e}quel, Tiphaine},
  journal={IMA J. Numer. Anal.},
  volume={38},
  number={1},
  pages={360--376},
  year={2018}
}

@article{kadau2016observing,
  title={Observing the {Rosensweig} instability of a quantum ferrofluid},
  author={Kadau, Holger and Schmitt, Matthias and Wenzel, Matthias and Wink, Clarissa and Maier, Thomas and Ferrier-Barbut, Igor and Pfau, Tilman},
  journal={Nature},
  volume={530},
  number={7589},
  pages={194--197},
  year={2016}
}

@article{schmitt2016self,
  title={Self-bound droplets of a dilute magnetic quantum liquid},
  author={Schmitt, Matthias and Wenzel, Matthias and B{\"o}ttcher, Fabian and Ferrier-Barbut, Igor and Pfau, Tilman},
  journal={Nature},
  volume={539},
  number={7628},
  pages={259--262},
  year={2016}
}

@article{chomaz2016quantum,
  title={Quantum-fluctuation-driven crossover from a dilute {Bose}-{Einstein} condensate to a macrodroplet in a dipolar quantum fluid},
  author={Chomaz, L and Baier, S and Petter, D and Mark, MJ and W{\"a}chtler, F and Santos, Luis and Ferlaino, F},
  journal={Phys. Rev. X},
  volume={6},
  number={4},
  pages={041039},
  year={2016}
}

@article{ferrier2016liquid,
  title={Liquid quantum droplets of ultracold magnetic atoms},
  author={Ferrier-Barbut, Igor and Schmitt, Matthias and Wenzel, Matthias and Kadau, Holger and Pfau, Tilman},
  journal={J. Phys. B: At. Mol. Opt. Phys.},
  volume={49},
  number={21},
  pages={214004},
  year={2016}
}

@article{wenzel2017striped,
  title={Striped states in a many-body system of tilted dipoles},
  author={Wenzel, Matthias and B{\"o}ttcher, Fabian and Langen, Tim and Ferrier-Barbut, Igor and Pfau, Tilman},
  journal={Phys. Rev. A},
  volume={96},
  number={5},
  pages={053630},
  year={2017}
}

@article{guo2021lee,
  title={{Lee}-{Huang}-{Yang} effects in the ultracold mixture of ${}^{23}\mathrm{Na}$ and ${}^{87}\mathrm{Rb}$ with attractive interspecies interactions},
  author={Guo, Zhichao and Jia, Fan and Li, Lintao and Ma, Yinfeng and Hutson, Jeremy M and Cui, Xiaoling and Wang, Dajun},
  journal={Phys. Rev. Research},
  volume={3},
  number={3},
  pages={033247},
  year={2021}
}

@article{baillie2017collective,
  title={Collective excitations of self-bound droplets of a dipolar quantum fluid},
  author={Baillie, D and Wilson, RM and Blakie, PB},
  journal={Phys. Rev. Lett.},
  volume={119},
  number={25},
  pages={255302},
  year={2017}
}

@article{tylutki2020collective,
  title={Collective excitations of a one-dimensional quantum droplet},
  author={Tylutki, Marek and Astrakharchik, Grigori E and Malomed, Boris A and Petrov, Dmitry S},
  journal={Phys. Rev. A},
  volume={101},
  number={5},
  pages={051601},
  year={2020}
}

@article{bao2004computing,
  title={Computing the ground state solution of {Bose}--{Einstein} condensates by a normalized gradient flow},
  author={Bao, Weizhu and Du, Qiang},
  journal={SIAM J. Sci. Comput.},
  volume={25},
  number={5},
  pages={1674--1697},
  year={2004}
}

@article{Luo2021DCDSB,
  title = {On 3d dipolar {Bose}-{Einstein} condensates involving quantum fluctuations and three-body interactions},
  author = {Yongming Luo and Athanasios Stylianou},
  journal = {Discrete Contin. Dyn. Syst. Ser. B},
  volume = {26},
  number = {6},
  pages = {3455-3477},
  year = {2021}
}

@book{lieb2005mathematics,
  title={The mathematics of the {Bose} gas and its condensation},
  author={Lieb, Elliott H and Solovej, Jan Philip and Seiringer, Robert and Yngvason, Jakob},
  year={2005},
  publisher={Springer}
}
\end{document}